\documentclass[12pt]{article}
\usepackage{amsmath,amsthm,amssymb}
\usepackage{algorithm,algorithmic}
\usepackage[title]{appendix}
\usepackage{color}
\usepackage{booktabs}
\usepackage{chngcntr}
\usepackage{comment}
\usepackage{dsfont}
\usepackage{endnotes}
\usepackage{enumerate}
\usepackage{float}
\usepackage[bottom]{footmisc}
\usepackage{geometry}
\usepackage{graphicx}
\usepackage{hyperref}
\usepackage{indentfirst}
\usepackage[mathlines]{lineno}
\usepackage{longtable}
\usepackage{lscape}
\usepackage{mathabx}
\usepackage{makecell}
\usepackage{multibib}
\usepackage{multirow}
\usepackage{natbib}
\usepackage{pifont}
\usepackage{setspace}
\usepackage{subcaption}
\usepackage{threeparttable}
\usepackage{titling,titlesec}
\usepackage{verbatim}
\bibpunct{(}{)}{;}{a}{,}{,}
\usepackage[title]{appendix}

\newcommand{\bfm}[1]{\ensuremath{\mathbf{#1}}}

\def\bd{\bfm d}

   \def\bM{\bfm M}

   \def\bP{\bfm P}  
     
\def\br{\bfm r}     \def\RR{\mathbb{R}}

\def\bv{\bfm v}     
\def\bw{\bfm w}   \def\bW{\bfm W}  
\def\bx{\bfm x}   \def\bX{\bfm X}  
\def\by{\bfm y}     
\def\bz{\bfm z}     

\def\calA{{\cal  A}}

\def\calD{{\cal  D}}

\def\calI{{\cal  I}}

\def\bzero{\bfm 0}
\newcommand{\cmark}{\ding{52}}%
\newcommand{\xmark}{\ding{55}}%
\newcommand{\bfsym}[1]{\ensuremath{\boldsymbol{#1}}}

\def\bdelta{\bfsym {\delta}}           \def\bDelta {\bfsym {\Delta}}
              
\def\bmu{\bfsym {\mu}}

\def\bsigma{\bfsym \sigma}             \def\bSigma{\bfsym \Sigma}

\def\E{\mathrm{E}}
\def\var{\text{Var}}

\def\corr{\text{corr}}

\def\prob{\mathrm{Pr}}

\renewcommand{\hat}{\widehat}

\usepackage{array}
\newcolumntype{L}[1]{>{\raggedright\let\newline\\\arraybackslash\hspace{0pt}}m{#1}}
\newcolumntype{C}[1]{>{\centering\let\newline\\\arraybackslash\hspace{0pt}}m{#1}}
\newcolumntype{R}[1]{>{\raggedleft\let\newline\\\arraybackslash\hspace{0pt}}m{#1}}

\floatstyle{ruled}
\newfloat{algorithm}{tbhp}{loa}
\floatname{algorithm}{Algorithm}
\newcites{New}{References}

\titleformat{\section}{\normalfont\Large\bfseries}{\thesection}{0.5em}{}
\titlespacing*{\section} {0pt}{5pt}{3pt}
\titlespacing*{\subsection} {0pt}{5pt}{2pt}
\numberwithin{equation}{section}
\theoremstyle{plain}
\newtheorem{theorem}{Theorem}

\newtheorem{lemma}{Lemma}
\newtheorem{remark}{Remark}

\allowdisplaybreaks[4]
\begin{document}
\title{Multiple-Splitting Projection Test for High-Dimensional Mean Vectors \vspace{-0.6ex}}
\author{
Wanjun Liu$^1$, Xiufan Yu$^2$ and Runze Li$^3$ \\ $^1$LinkedIn,  $^2$University of Notre Dame, and $^3$Pennsylvania State University
}

\date{
April 17, 2022
}
\maketitle{}

\pagestyle{plain}

\vspace{-2ex}
\begin{abstract}
We propose a multiple-splitting projection test (MPT) for one-sample mean vectors in high-dimensional settings. The idea of projection test is to project high-dimensional samples to a 1-dimensional space using an optimal projection direction such that traditional tests can be carried out with projected samples. However, estimation of the optimal projection direction has not been systematically studied in the literature. In this work, we bridge the gap by proposing a consistent estimation via regularized quadratic optimization. To retain type I error rate, we adopt a data-splitting strategy when constructing test statistics. To mitigate the power loss due to data-splitting, we further propose a test via multiple splits to enhance the testing power. We show that the $p$-values resulted from multiple splits are exchangeable.  Unlike existing methods which tend to conservatively combine dependent $p$-values, we develop an exact level $\alpha$ test that explicitly utilizes the exchangeability structure to achieve better power. Numerical studies show that the proposed test well retains the type I error rate and is more powerful than state-of-the-art tests.
\end{abstract}

\noindent {\textbf{Key Words}:} Exchangeable $p$-values; High-dimensional mean tests; Multiple data-splitting; Optimal projection direction; Regularized quadratic optimization.

\newpage

\section{Introduction} \label{sec:intro}
Hypothesis testing on mean vectors is a fundamental problem in statistical inference theory and attracts considerable interest in numerous scientific applications. For example, neuroscientists make inferences on the average signals of fMRI data to monitor brain activities and diagnose abnormal tissues \citep{ginestet2017hypothesis}. Geneticists analyze gene expression levels to understand the mechanism of how genes are related to diseases \citep{wang2015high}.
In these applications, the data dimension $p$ is typically comparable with or much larger than sample size $n$, making traditional tests ineffective or practically infeasible.
In this work, we study the problem of testing whether a population mean $\bmu$ equals to some known vector $\bmu_0$ under high-dimensional regime where $p > n$. Without loss of generality, we set $\bmu_0 = \bzero$ throughout the paper. To formally formulate the problem, let $\bX = (\bx_1, \dots, \bx_n)^\top$ be a random sample from a $p$-dimensional population $\bx$ with mean $\bmu$ and covariance $\bSigma$. Of interest is to test
\begin{equation}
    H_0: \bmu = \bfm 0 \quad \text{versus} \quad H_1: \bmu \neq \bfm 0.
\label{eqn:test_mu0}
\end{equation}
The Hotelling's $T^2$ test has been well studied when $p<n$ and $p$ is fixed. As $p$ exceeds $n$, the sample covariance matrix becomes singular and hence $T^2$ is not well-defined. Even in the case $p<n$, the testing power of $T^2$ is largely defective if $p/n\rightarrow c\in(0,1)$ \citep{bai1996effect}.

Three types of tests have been developed in efforts to handle the high-dimensional challenge. The first type is quadratic-form test, which replaces the singular sample covariance matrix with an invertible matrix (e.g., identity matrix) \citep{bai1996effect,chen2011regularized,chen2010two}. These tests tend to neglect the dependence among covariates and may suffer from low power when covariates are strongly correlated. The second type is known as extreme-type test, which utilizes the extreme value of a sequence of marginal test statistics, see \cite{cai2014two,zhong2013tests}. Such extreme-type statistics typically converge to some extreme value distribution and are generally disadvantaged by slow convergence, making it hard to control the type I error when $n$ is small. The third type is projection test \citep{lopes2011more,li2015projection,liu2020projection,runze2021linear}, which maps the high-dimensional sample to a low-dimensional space, and subsequently applies traditional methods (e.g., Hotelling's $T^2$) to the projected sample. Intuitively, the projection procedure seeks to transform the data in such a way that the dimension is reduced, while the statistical distance between $H_0$ and $H_1$ is mostly preserved through the transformed distributions.

Recently, \cite{li2015projection} proved that the optimal choice of projection direction is $\bSigma^{-1}\bmu$. To facilitate a data-driven decision regarding the projection direction, \cite{li2015projection} also proposed a projection test based on a data-splitting procedure, i.e., half of the sample is employed to estimate the optimal projection direction, while the other half is used to perform the test.
However, there are two main drawbacks with this data-splitting projection test. First, a ridge-type estimator is used to estimate the projection direction. Their power analysis relies on the assumption that the ridge-type estimator is consistent, which is no longer true in high-dimensional settings.
Secondly, the single data-splitting procedure is often criticized as only half of the sample is used to perform the test, which inevitably results in power loss.
These two drawbacks actually reveal two existing unsolved issues with the projection test based on a data-splitting procedure:
\smallskip

1. How to estimate the optimal projection direction with statistical guarantee?

2. How to mitigate the power loss caused by the data-splitting procedure?
\smallskip

In this paper, we propose a multiple-splitting projection test for high-dimensional mean vectors. Our proposed test addresses the aforementioned issues in the following two ways: (1) the optimal projection is estimated via a regularized quadratic optimization such that a consistent estimator is obtained; and (2) a multiple data-splitting procedure is proposed to improve the testing power. The main contributions can be summarized in three folds.


First, we propose a consistent estimation of the optimal projection direction via nonconvex regularized quadratic programming. Non-asymptotic error bounds are established, which hold for \emph{all stationary} points with high probability. In other words, we do not need to solve the global solution to the nonconvex optimization problem as any stationary point has desirable statistical guarantee.

Second, we prove that $p$-values constructed from a multiple data-splitting procedure are \emph{exchangeable}.
Furthermore, we generalize the exchangeability of $p$-values proposition to a more general permutation framework. As an extension, the methodology proposed in this work can be further applied to many statistical inference problems.

Third, an \emph{exact} level $\alpha$ test is proposed to combine multiple $p$-values which explicitly utilizes the exchangeability of these $p$-values. Such exchangeability is often neglected in traditional combination approaches. By doing so, our test is more powerful than the single-splitting test as well as existing combination approaches.
To the best of our knowledge, this is the first work that exploits the exchangeability of $p$-values and utilizes such exchangeability in developing high-dimensional hypothesis testing.

The rest of this paper is organized as follows. 
In Section \ref{sec:estimation}, we introduce a new estimation of the optimal projection direction via regularized quadratic programming. In Section \ref{sec:mpt}, we investigate the dependency structure of $p$-values resulted from a multiple-splitting procedure and propose an exact level $\alpha$ multiple-splitting projection test. In Section \ref{sec:numerical}, we conduct numerical studies to compare the proposed MPT with existing tests as well as other $p$-value combination methods. We conclude this paper with discussion on potential applications of this multiple-splitting framework to other statistical inference problems in Section \ref{sec:discuss}.

\section{Estimation of Optimal Projection Direction}\label{sec:estimation}
In this section, we introduce a consistent estimation of the optimal projection direction for projection tests. Section \ref{subsec:background} provides a brief introduction to projection tests. Section \ref{subsec:estimation} presents the estimator as a  stationary point of a regularized quadratic optimization problem and establishes its non-asymptotic error bounds.

\subsection{Background on Projection Tests}
\label{subsec:background}
The idea of projection test is to project the high-dimensional vector $\bx \in \RR^p$ onto a space of low dimension such that traditional tests can be applied. Let $\bP$ be a $p\times q$ full column-rank projection matrix (or vector if $q=1$) with $q<n$ and define $\by_i = \bP^\top\bx_i \in \RR^q, i=1,\dots,n$. Under $H_0$, $\E(\by_i)=\bfm 0$ and Hotelling's $T^2$ test can be applied to the $q$-dimensional projected sample $\by_i$'s,
\begin{equation*}
    T^2_{\bP} = n \widebar\bx^\top\bP(\bP^\top\widehat\bSigma\bP)^{-1}\bP^\top\widebar\bx,
\end{equation*}
where $\widebar\bx$ and $\widehat\bSigma$ are the sample mean and sample covariance matrix. Under $H_0$, $T^2_\bP$ converges to $\chi^2_q$ distribution as $n\rightarrow \infty$.

The projection test pivots the attention to the question on how to effectively construct the projection matrix $\bP$. Various approaches have been developed with respect to different choices of $\bP$. A data-dependent method was proposed in \cite{lauter1996exact} by setting $\bP=\bd$, where $\bd$ is a $p\times 1$ vector depending on data only through $\bX ^\top\bX$. 
\cite{lopes2011more} proposed a random projection test in which the entries in $\bP$ are randomly drawn from standard normal distribution.
More recently, \cite{li2015projection} proved that under normality assumption, the optimal choice $q$ is 1 and the optimal projection direction is of the form $\bP = \bSigma^{-1}\bmu$ in the sense that the power of $T^2_{\bP}$ is maximized. For non-Gaussian samples, the direction $\bP = \bSigma^{-1}\bmu$ is still asymptotically optimal as long as the sample mean of projected sample is asymptotically normal.

The estimation of the optimal projection direction has not been systematically studied yet, leaving a gap between theory and practice for projection tests. In what follows, we propose a new estimating procedure such that a consistent estimator is obtained.

\textbf{Notations:} Before proceeding, we first set up some notations. For a vector $\bv = (v_j)_{j=1}^p\in \RR^p$, let $\|\bv\|_k$ be its $\ell_k$ norm, $k=1,2$. Its $\ell_0$ norm $\|\bv\|_0$ is the number of nonzero entries in $\bv$ and $\ell_\infty$ norm is $\|\bv\|_\infty = \max |v_j|$. For a matrix $\bM = (m_{ij})_{i,j=1}^p \in \RR^{p\times p}$, its elementwise $\ell_\infty$ norm is $\|\bM\|_{\max} = \max |m_{ij}|$. For a set $\calD$, $|\calD|$ denotes its cardinality. We use $a \vee b$ to denote the larger one of $a$ and $b$.

\subsection{Estimation via Regularized Quadratic Optimization}
\label{subsec:estimation}

In this subsection, we aim to bridge the gap between theoretical analysis and practical implementation regarding the optimal projection direction. The empirical performance of a data-driven projection test relies heavily on the estimation accuracy of the optimal projection direction. However, in high-dimensional settings, there is no statistical guarantee for the ridge-type estimator introduced in \cite{li2015projection}.

We propose a new consistent estimator to improve the test performance with the assumption that $\bw^\star$ is sparse. Observing that $\bSigma^{-1}\bmu$ is the minimizer of $\frac{1}{2} \bw^\top\bSigma\bw - \bmu^\top\bw$, we propose to estimate $\bw^\star = \bSigma^{-1}\bmu$ using the following regularized quadratic optimization
\begin{equation}
\underset{\bw}{\text{minimize}} \ \frac{1}{2} \bw^\top \hat\bSigma\bw - \bar{\bx}^\top \bw + P_{\lambda}(\bw),
\label{eqn:estimation}
\end{equation}
where $P_\lambda(\bw) = \sum_{j=1}^p P_\lambda(w_j)$ is a penalty function satisfying the following conditions
\begin{enumerate}[\quad (i)]
    \item $P_\lambda(0) = 0$ and $P_\lambda(t)$ is symmetric around 0,
    \item $P_\lambda(t)$ is differentiable for $t\neq 0$ and $\lim_{t \rightarrow 0^+} P'_\lambda(t) = \lambda$,
    \item $P_\lambda(t)$ is a non-decreasing function on $t \in [0, \infty)$,
    \item $P_\lambda(t)/t$ is a non-increasing function on $t \in [0, \infty)$,
    \item There exists $\gamma>0$ such that $P_\lambda(t) + \frac{\gamma}{2}t^2$ is convex.
\end{enumerate}
Such conditions on $P_\lambda$ are mild \citep{loh2015regularized} and are satisfied by a wide variety of penalties including the Lasso \citep{tibshirani1996regression} and nonconvex regularizers such as the SCAD \citep{fan2001variable}, and the MCP \citep{zhang2010nearly}. We further assume that the sample covariance matrix $\hat\bSigma$ satisfies the following restricted strong convexity (RSC) condition,
\begin{equation}
    \bDelta^\top\hat\bSigma\bDelta \geq
    \nu \|\bDelta\|_2^2 - \tau\sqrt{\frac{\log p}{n}} \|\bDelta\|_1 \ \text{for} \ \bDelta \in \RR^p \ \text{and} \ \|\bDelta\|_1 \geq 1,
\label{eqn:rsc}
\end{equation}
where $\nu>0$ is a strictly positive constant and $\tau \geq 0$ is a non-negative constant. When $p < n$, $\hat\bSigma$ is positive definite, one can set $\tau=0$ and $\nu$ be the smallest eigenvalue of $\hat\bSigma$.
In the high-dimensional setting where $p > n$, $\hat\bSigma$ is semi-positive definite and $\bDelta^\top\hat\bSigma\bDelta \geq 0$ for all $\bDelta \in \RR^p$. Thus the RSC condition \eqref{eqn:rsc} holds trivially for $\{ \bDelta: {\|\bDelta\|_1}/{\|\bDelta\|_2^2} > c\}$ with $c = \frac{\nu}{\tau}\sqrt{\frac{n}{\log p}}$. As a result, we only require the RSC condition to hold in the set $\{ \bDelta: {\|\bDelta\|_1}/{\|\bDelta\|_2^2} \leq c, \|\bDelta\|_1\geq 1 \}$. The RSC condition \eqref{eqn:rsc} is imposed on $\hat\bSigma$ only for $\|\bDelta\|_1 \geq 1$, and it turns out the condition actually holds for all $\bDelta\in\RR^p$, see Lemma \ref{lem:full_rsc}. Such RSC-type condition is widely used in establishing the non-asymptotic error bounds in high-dimensional statistics and is satisfied with high probability under sub-Gaussianity assumption \citep{agarwal2012fast,loh2015regularized,loh2017support}. Alternatively, the RSC-type condition can also be replaced by a similar condition known as restricted eigenvalue (RE) condition \citep{bickel2009simultaneous,van2009conditions}.

It is quite challenging to obtain the global solution to the optimization problem  \eqref{eqn:estimation} if a nonconvex penalty $P_\lambda$ is used. Instead of searching for the global solution, we establish the non-asymptotic error bounds for any stationary point $\hat\bw$ that satisfies the following first-order condition,
\begin{equation}
     \hat\bSigma\hat\bw - \bar\bx + \nabla P_\lambda(\hat\bw) = \bfm 0,
\label{con:first_order}
\end{equation}
where $\nabla P_\lambda$ denotes the sub-gradient of $P_\lambda$. The condition (\ref{con:first_order}) is a necessary condition for $\hat\bw$ to achieve a local minimum. Therefore the set of $\hat\bw$ satisfying (\ref{con:first_order}) includes all local minimizers as well as the global one.

Lots of efficient algorithms have been developed to attain stationary points even when the objective function is nonconvex. These algorithms include local linear approximation \citep{fan2014strong,wang2013calibrating,zou2008one}, composite gradient descent method \citep{loh2015regularized,nesterov2013gradient}, and proximal-gradient method \citep{wang2014optimal}.
In practice, we may choose the tuning parameter $\lambda$ in the penalty function by cross-validation or the high-dimensional BIC criterion proposed in \cite{wang2013calibrating}.
We impose the following conditions,
\begin{description}
  \item[(C1)] $\bx_1, \dots, \bx_n$ are identically and independently distributed sub-Gaussian vectors.
  \item[(C2)] The sample covariance matrix $\hat{\bSigma}$ satisfies the RSC condition in (\ref{eqn:rsc}) with $3 \gamma \leq 4 \nu$.
  \item[(C3)] There exists constant $C_1>0$ such that $\|\bw^\star\|_1 \leq C_1$.
\end{description}

\begin{remark}
Condition (C3) is posited to ensure a good estimation of $\bw^\star$. By the definition of $\bw^\star$, $\hat\bSigma\bw^\star$ should be somewhat close to $\bmu$. Note that $\| \hat\bSigma\bw^\star  - \bmu\|_\infty = \| \hat\bSigma\bw^\star  - \bSigma\bw^\star\|_\infty \leq \|\hat\bSigma - \bSigma\|_{\max} \cdot \|\bw^\star\|_1$. A diverging $\|\bw^\star\|_1$ would amplify the estimation error of $\hat\bSigma$.
\end{remark}

\noindent The following theorem establishes the $\ell_1$ and $\ell_2$ error bounds for all stationary points $\hat\bw$ under the alternative hypothesis.
\begin{theorem}
Suppose conditions (C1)-(C3) hold. Let $\hat{\bw}$ be any stationary point of the problem (\ref{eqn:estimation}) with $\lambda = C\sqrt{\log p/n}$ for some large constant $C$. Then under $H_1$ (i.e., $\bw^\star \neq \bfm 0$), with probability at least $1-cp^{-1}$ for some absolute constant $c$, we have
\begin{equation*}
\| \hat{\bw}  - {\bw}^\star\|_1 = O\left(s\sqrt{\frac{\log p}{n}}\right) \quad and \quad \|\hat{\bw}  - {\bw}^\star\|_2 = O\left(\sqrt{\frac{s\log p}{n}}\right),
\end{equation*}
where $s = \|\bw^\star\|_0$ is the number of nonzero entries in $\bw^{\star}$.
\label{thm:error_bound}
\end{theorem}
\begin{remark}
Though inspired by \cite{loh2015regularized}, we would like to clarify the difference between Theorem \ref{thm:error_bound} and the results in \cite{loh2015regularized}. The optimization problem in \cite{loh2015regularized} requires an additional constraint $\|\bw\|_1 \leq R$ for some tuning parameter $R$ to ensure $\|\hat\bw\|_1$ is bounded by $R$. $R$ needs to be chosen carefully such that $\bw^\star$ is feasible and both the penalty parameter $\lambda$ and sample size $n$ also depend on $R$. However, how to choose $R$ is not clear in practice. In our work, we modify the RSC condition by substituting $\|\bDelta\|^2_2$ for $\|\bDelta\|_2$ in the RSC condition (\ref{eqn:rsc}) so that the constraint $\|\hat\bw\|_1 \leq R$ is no longer needed.
\end{remark}

Note that the error bounds in Theorem~\ref{thm:error_bound} hold for all stationary points. In other words, any local solution is guaranteed with desirable statistical accuracy and a global one is unnecessary if it is too challenging to achieve.
Theorem \ref{thm:error_bound} implies $\hat\bw$ is a consistent estimator under $H_1$ if $\bw^\star$ is sparse (more precisely, $\sqrt{s \log p/n}\rightarrow 0$). Such consistency under $H_0$ is not guaranteed as there is no signal in the true parameter $\bw^\star = \bfm 0$. Fortunately, with the data-splitting technique, we will see in Section \ref{sec:mpt} that the size of the proposed projection test is always controlled regardless of the consistency of the estimator $\hat\bw$.

\section{Data-Splitting Based Projection Test}
\label{sec:mpt}
In this section, we present full methodological details of our proposed multiple-splitting projection test (MPT) together with its theoretical properties. After carefully studying the dependency of $p$-values resulted from a multiple-splitting procedure, we introduce a new combination framework that makes use of the dependency structure. Section \ref{subsec:SPT} demonstrates a single-splitting projection test using the estimator introduced in the previous section. Section \ref{subsec:exchangebility} studies the exchangeability of $p$-values. Section \ref{subsec:traditionalCombination} provides a brief overview of traditional approaches for combining multiple $p$-values. Section \ref{subsec:MPT} formally presents our combination framework and the proposed MPT.

\subsection{Single-Splitting Projection Test (SPT)}
\label{subsec:SPT}

Data-splitting technique has a long history in statistical applications
and remains attractive in modern statistics \citep{wasserman2009high,barber2019knockoff}.
We begin with one single data-splitting. Let $\calD = \{\bx_1, \dots, \bx_n\}$ denote the set of full sample and we partition the full sample into two disjoint sets $\calD_1 = \{\bx_1, \dots, \bx_{n_1}\}$ and $\calD_2 = \{\bx_{n_1+1}, \dots, \bx_n\}$ with $|\calD_1|=n_1$ and $|\calD_2|=n_2=n-n_1$. The idea is to use $\calD_1$ to estimate the optimal projection direction while use $\calD_2$ to conduct the test with projected sample.
To be more specific, we estimate the optimal projection direction $\bw^\star$ using a stationary point $\widehat\bw$ of the following regularized quadratic optimization problem
\begin{equation}
  \underset{\bw}{\text{minimize}} \ \frac{1}{2} \bw^\top \hat\bSigma_1\bw - \bar{\bx}_1^\top \bw + P_{\lambda}(\bw),
\label{eqn:estimation_D1}
\end{equation}
where $\bar{\bx}_1$ and $\hat\bSigma_1$ are sample mean and sample covariance matrix computed from $\calD_1$.
Then we project the observations in $\calD_2$ to a 1-dimensional space as follows: $y_i = \hat\bw^\top \bx_{i}, i=n_1+1,\dots, n$. The one-sample $t$-test is readily applied to the projected sample and the resulting test statistic is
\begin{equation}
    T_{\hat\bw}= \sqrt{n_2}\widebar{y}/s_y,
\end{equation}
where $\widebar y$ and $s_y^2$ are the sample mean and sample variance of $\{ y_{n_1+1}, \dots, y_n \}$. Due to the data-splitting, the estimator $\hat\bw$ is independent of $\calD_2$. As a result, the test $T_{\widehat\bw}$ is an exact one-sample $t$-test if $\bx_i$'s are normally distributed, and the $p$-value of test $T_{\hat\bw}$ is given by
\begin{equation}
    p(T_{\hat\bw}) = 2\left( 1-G_{n_2-1}(|T_{\widehat\bw}|) \right),
\label{eqn:p_value}
\end{equation}
where $G_{n_2-1}$ is the cdf of $t_{n_2-1}$ distribution. Without normality assumption, $T_{\widehat\bw}$ has an asymptotic standard normal distribution, and the $p$-value is given by $p(T_{\hat\bw}) = 2\left( 1-\Phi(|T_{\widehat\bw}|) \right). $
We reject the null hypothesis at significance level $\alpha$ whenever $p(T_{\hat\bw}) < \alpha$.

We refer to $T_{\widehat\bw}$ as the single-splitting projection test (SPT). Ideally, one would like to use full sample to estimate $\bw^\star$ and use full sample to perform the test, which makes the limiting distribution challenging to derive since the projection $\hat\bw$ and sample are dependent.
Thanks to the data-splitting procedure, an exact $t$-test can be achieved as $\hat\bw$ is independent of $\calD_2$. Furthermore, we would like to point out that the size of the SPT is well controlled regardless of how $\hat\bw$ is estimated, but a consistent $\hat\bw$ ensures high power under the alternative. The following theorem demonstrates the asymptotic power of the SPT.
\begin{theorem} \label{thm:SPT_power}
Suppose that the conditions in Theorem \ref{thm:error_bound} hold. Further assume \\ $(1 \vee \|\bmu\|_\infty)s\sqrt{\log p/n} \rightarrow 0$ and $n_2/n \rightarrow \kappa \in (0,1)$, where $n_2$ is the sample size of $\calD_2$. Let $\zeta = \bmu^\top\bSigma^{-1}\bmu$ and $z_{\alpha/2}$ be the upper $\alpha/2$ quantile of $N(0,1)$, then the asymptotic power of the proposed SPT at a given significance level $\alpha$ is
\begin{equation*}
   \beta(T_{\widehat\bw}) = \Phi(-z_{\alpha/2} + \sqrt{n\kappa\zeta}).
\end{equation*}
\end{theorem}

\begin{remark}
The term $\zeta$ can be interpreted as the signal strength of alternative hypothesis. As long as $n \bmu^\top\bSigma^{-1}\bmu \rightarrow \infty$, the proposed SPT has asymptotic power approaching 1. Under local alternative $\bmu = \bdelta/\sqrt{n}$ for some fixed $\bdelta \neq \bfm 0$, the asymptotic power of the SPT is $\Phi(-z_{\alpha/2} + \sqrt{\kappa\bdelta^\top\bSigma^{-1}\bdelta})$.
To achieve high power empirically, we adopt the same strategy as in  \cite{li2015projection} and recommend to take $n_2 = \lfloor \kappa n \rfloor$ with $\kappa \in [0.4,0.6]$.
\end{remark}

\subsection{Exchangeability of $p$-values}
\label{subsec:exchangebility}

To compensate the power loss due to the  data-splitting procedure, we consider a multiple-splitting procedure (formally presented in Section \ref{subsec:MPT}), which repeats the data-splitting multiple times and aggregates all the information in $p$-values to make inferences on $H_0$. More specifically, we consider $m$ times of data-splitting for some fixed integer $m$. Let $\pi_k$, $k=1, \dots, m$, be a random permutation of $\{1, \dots, n\}$. Accordingly, let $\calD^{\pi_k} = \{\bx_{\pi_k(1)}, \dots, \bx_{\pi_k(n)}\}$ denote the permutated sample. For each $k = 1, \dots, m$, we apply the SPT to $\calD^{\pi_k}$ and obtain the $p$-value $p_k$ according to \eqref{eqn:p_value}.
Before proceeding with the combination of $p$-values, we first investigate the dependence structure of these $p$-values. The following theorem establishes the exchangeability among the $p$-values.
\begin{theorem} \label{thm:exchangeable}
The $p$-values $(p_1, p_2, \dots, p_m)$ resulted from the multiple-splitting procedure are \\ exchangeable, i.e., $(p_1, \dots, p_m) \overset{d}{=} (p_{\pi(1)}, \dots, p_{\pi(m)})$ for any $\pi$, a permutation of $\{1,\dots,m\}$.
\end{theorem}

The exchangeability of $p$-values holds no matter it is under $H_0$ or $H_1$.
We would like to point out that such exchangeability structure of $p$-values holds for a general permutation framework. In many statistical problems, with the technique of data-splitting, the first half sample $\calD_1^{\pi_k}$ can be used to learn the underlying model (e.g., parameter estimation, variable selection). We denote the acquired knowledge by $f(\calD_1^{\pi_k})$. Then together with the second half sample $\calD_2^{\pi_k}$, a $p$-value (or some other statistic) can be derived for some specific inference problem $p_k = g(f(\calD_1^{\pi_k}), \calD_2^{\pi_k} )$. With fixed mappings $f$ and $g$, it can be shown that the $p_k$'s are also exchangeable. For instance, let us consider the high-dimensional linear regression problem and of interest is to test whether some coefficient, say $\beta_j$, is 0 or not. With data-splitting, one can use $\calD_1^{\pi_k}$ to select a set of important covariates $\hat\calA$ such that $|\hat\calA|<n-1$. Then we can fit ordinary least squares with the  covariate set $\hat\calA \cup \{j\}$ and obtain the $p$-value $p_k$ of a test regarding whether $\beta_j=0$. Since the $p$-values are exchangeable, the MPT introduced in Section \ref{subsec:MPT} is readily to combine the $p$-values.

In fact, such exchangeability holds even without data-splitting. The key of exchangeability lies in that conditioning on the dataset $\calD$, its $m$ permutations $\calD^{\pi_1}, \dots, \calD^{\pi_m} $ are independent from each other. The following theorem generalizes the results in Theorem \ref{thm:exchangeable}.

\begin{theorem}\label{thm:general_exchangeable}
Let $\calD=\{\bx_1,\dots,\bx_n\}$ be a random sample and $\pi_1, \cdots, \pi_m$ be $m$ permutations of $\{1, \cdots, n\}$. $\calD^{\pi_1}, \dots, \calD^{\pi_m} $ denote the $m$ permutated samples of $\calD$. Let $g$ be a mapping from $\calD^{\pi_k}$ to a statistic: $T_k = g(\calD^{\pi_k})$, then $T_1, \dots, T_m$ are exchangeable.
\end{theorem}

\subsection{Traditional Combination of $p$-values}
\label{subsec:traditionalCombination}
One popular strategy to enhance testing power is via the combination of $p$-values \citep{romano2019multiple,yu2019innovated,yu2020fisher,yu2022power}. In fact, 
combining multiple $p$-values from a set of hypothesis tests has been widely used in statistical literature. Let $p_1, \dots, p_m$ denote $m$ valid $p$-values. That is, under $H_0$,
\begin{equation*}
    \prob(p_k \leq u) = u,\ 0 < u < 1\quad  \text{for } k = 1,\dots, m.
\end{equation*}
A natural question is how we can decide whether $H_0$ should be rejected based on the $m$ $p$-values such that the type I error rate is still retained.

Classical approaches require independence assumption among $p$-values. Examples include the Fisher's method, 
the Pearson's method, 
the Stouffer's method, 
the Tippett's method, 
and many others.
In the meantime, significant efforts were made to combine dependent $p$-values.
\cite{ruschendorf1982random} proved that twice the average $p$-values remains a valid $p$-value and proposed an average-based combination test, which rejects $H_0$ at level $\alpha$ if the average of ${p_1,\dots, p_m}$ is less than or equal to $\alpha/2$.
\cite{romano2019multiple} introduced a quantile-based combination test. A special case is we reject $H_0$ at level $\alpha$ if the median of $p$-values is less than or equal to $\alpha/2$.
Under $H_0$, we know $Z_k = \Phi^{-1}(p_k) \sim N(0,1)$, where $\Phi(\cdot)$ is the cdf of $N(0,1)$. Assuming $(Z_1, \dots, Z_m)^\top$ follows a multivariate normal distribution, \cite{romano2019multiple} proposed a Z-average test based on the sample mean of $Z_k$'s, that rejects $H_0$ if $| \sum_{k=1}^m Z_k| \geq m z_{\alpha/2}$, where $z_{\alpha/2}$ is the upper $\alpha/2$-quantile of $N(0,1)$.
More recently, \cite{liu2020cauchy} introduced a new combination test based on the Cauchy transformation which is insensitive to dependencies among $p$-values. The test rejects $H_0$ at level $\alpha$ if $ \sum_{k=1}^m \tan\{(0.5-p_i)\pi\} \geq mc_{\alpha}$, where $c_{\alpha}$ is the upper $\alpha$-quantile of standard Cauchy distribution.


In general, these methods tend to be over-conservative in order to
control Type I  error rate without taking advantage of certain
dependence structure (e.g., exchangeability). This can be regarded
as a trade-off between potential size inflation and possible power
loss. The traditional combination methods generally ignore the
dependency structure, therefore tend to make unnecessarily large
comprise in testing power in order to retain a correct size for a
less favorable scenario. In Section 4.1, we use
simulation studies to compare the size and power of the proposed
MPT with those of traditional combination methods. On the one
hand, the numerical studies show that the size of MPT is slightly
below the level $\alpha=0.05$ while the size of traditional
combination methods are almost $0$. On the other hand, the power
of MPT is much higher than those traditional combination methods.
In summary, compared to the traditional ones, the proposed MPT is
less conservative in terms of testing size, and exhibits much
higher testing power.


\subsection{Multiple-Splitting Projection Test (MPT)}
\label{subsec:MPT}
Based on the $m$ exchangeable $p$-values $\{p_1, \dots, p_m\}$, the question is how we can make a decision on whether $H_0$ should be rejected or not. To improve upon the traditional methods, we propose a new framework to combine $p$-values obtained from multiple splits. The proposed framework takes advantage of the exchangeability structure among those $p$-values, as a result, achieving higher testing power than existing commonly used combination approaches.

\begin{figure}[!htb]
\centering
    \includegraphics[width=0.95\textwidth]{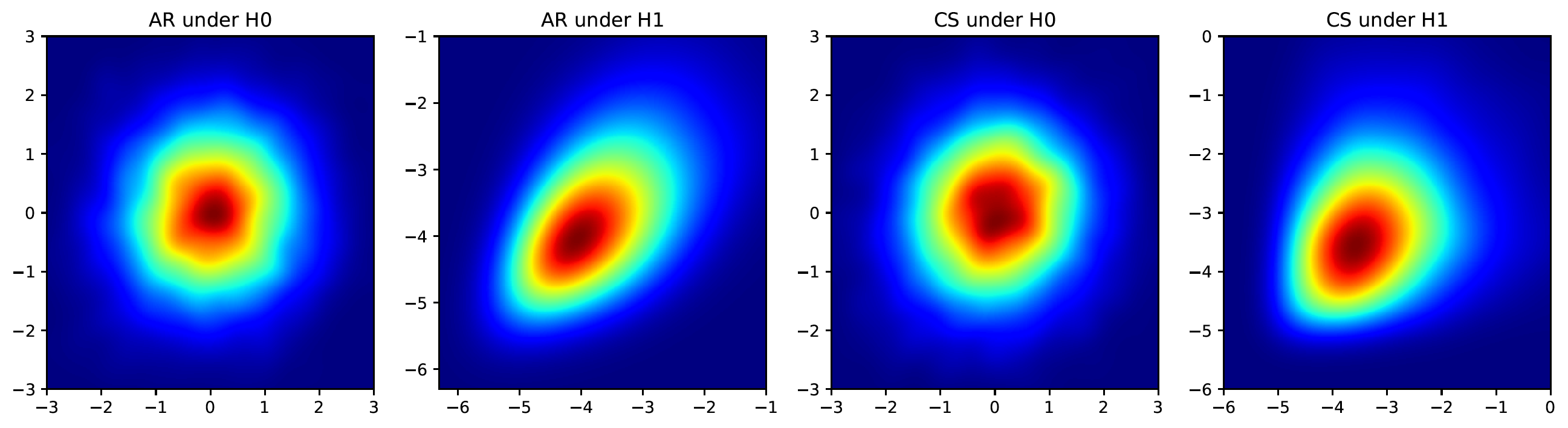}

\caption{Density plots of $(Z_1, Z_2)$ under $H_0$ and $H_1$ with autocorrelation (AR) covariance structure and compound symmetry (CS) covariance structure when $n=40$ and $p=1000$.}

\label{fig:density_4}
\end{figure}

Let $Z_k = \Phi^{-1}(p_k)$, $k=1,\dots, m$. Under $H_0$, $Z_1, \dots, Z_m$ are exchangeable standard normal random variables, with correlation $\corr(Z_i, Z_j) = \rho \geq 0$ for any pair $(i,j), i \neq j$ due to exchangeability.
Figure \ref{fig:density_4} depicts the density of $(Z_1, Z_2)$ with $m=2$ under different covariance structures (see Section \ref{sec:numerical} for detailed descriptions of simulation settings). It shows that $(Z_1, Z_2)$ are clearly exchangeable (symmetric). Under $H_0$, $(Z_1, Z_2)$ are approximately normally distributed centering at $(0,0)$. Under $H_1$, the joint distribution of $(Z_1, Z_2)$ is not normal and its center is far away from $(0,0)$.

Let $\widebar{Z}$ be the sample mean, then we have $\E(\bar Z)= 0$ and $\var(\bar Z) = (1 + (m-1)\rho)/m$ under $H_0$. If $(Z_1, \dots, Z_m)$ are jointly normally distributed, then the standardized statistic of $\widebar{Z}$, $M_{\rho} = {\bar Z}/{\sqrt{(1 + (m-1)\rho)/m}} \sim N(0,1).$
In general, by the central limit theorem for exchangeable random variables (e.g., see \cite{klass1987central}), we know
\begin{equation}\label{eqn:mrho}
    M_{\rho} = {\bar Z}/{\sqrt{(1 + (m-1)\rho)/m}} \overset{d}{\rightarrow} N(0,1).
\end{equation}
The correlation $\rho$ is typically unknown and needs to be estimated from the sample. Two approaches to estimate $\rho$ will be provided later in this subsection. Let $\hat \rho$ denote an estimator of $\rho$. We first present the asymptotic distribution of $M_{\widehat\rho}$ under $H_0$.
\begin{theorem}
\label{thm:Mhatrho}
Let $\hat \rho$ be a consistent estimator of $\rho > 0$.
Under $H_0$, we have as $m \rightarrow \infty$,
\begin{equation} \label{eq: Mhatrho}
    M_{\hat\rho} = {\bar Z}/{\sqrt{(1 + (m-1)\hat\rho)/m}} \overset{d}{\rightarrow} N(0,1).
\end{equation}

\end{theorem}

\begin{remark}
\label{rmk:slow_converge}
When $\rho=0$, the $p$-value combination is reduced to the independent case. Hence  $M_{\rho}$ converges to the standard normal distribution at rate $\sqrt{m}$ 
as $m \rightarrow \infty$. However, when $\rho \neq 0$, $M_{\widehat{\rho}}$ converges at a slower rate $m$ as the variance of $\bar Z$ is not asymptotically degenerate. Hypothesis testing based on such a slow convergence rate is more likely to fail in controlling the type I error in finite-sample performance.
\end{remark}

Remark \ref{rmk:slow_converge} suggests that the asymptotic distribution in \eqref{eq: Mhatrho} does not serve as a good cornerstone to test $H_0$.
The slow convergence rate and potential size inflation become major concerns for practitioners. In practice, one may conduct a large number of splits which bring in extra computational burden.
This motivates us to seek an exact level $\alpha$ test to ensure the finite-sample performance. Let $c({\rho}, m, \alpha/2)$ be the upper $\alpha/2$ quantile of the distribution of $M_{\hat\rho}$ and we reject $H_0$ if $|M_{\widehat{\rho}}| > c({\rho}, m, \alpha/2)$. Given $\hat\rho$, the exact distribution of $M_{\widehat{\rho}}$ depends on $\rho$ and is very difficult to derive, so is $c({\rho}, m, \alpha/2)$. Instead, we use the critical value $c(m, \alpha)$ that is chosen against the least favorable $\rho$ such that type I error is controlled regardless $\rho$, i.e., $c(m, \alpha/2) = \sup_{\rho \in (0,1)} c({\rho}, m, \alpha/2)$. Then we reject $H_0$ at level $\alpha$ if
\begin{equation} \label{eqn: MPT}
    |M_{\widehat{\rho}}| > c(m, \alpha/2).
\end{equation}
We refer to the test \eqref{eqn: MPT} as multiple-splitting projection test (MPT) and summarize full methodological details in Algorithm \ref{algo: MPT}. Note that the critical value $c(m, \alpha/2)$ depends on the way you estimate $\rho$. With the choice $c(m, \alpha/2)$, the MPT is still an exact level $\alpha$ test but no longer a size $\alpha$ test.

\begin{algorithm} [H]
\caption{Multiple-splitting Projection Test (MPT)}\label{algo: MPT}
\small
\begin{algorithmic}[1]
\STATE \textbf{Input:} dataset $\calD$, the number of splits $m$, $n_1$, and significance level $\alpha$
\STATE \textbf{Step 1:} randomly generate $m$ permutations of $\{1, \dots, n\}$, denoted by $\pi_k$, $k=1,\dots, m$
\STATE \textbf{Step 2:} obtain multiple $p$-values 
\FOR{ $k=1$ \TO $m$ }
\STATE (1) partition the permuted sample $\calD^{\pi_k}$ into $\calD^{\pi_k}_1$ and $\calD^{\pi_k}_2$
and obtain $\bar\bx_1^{k}$, $\widehat{\bSigma}^{k}_1$ from $\calD^{\pi_k}_1$
\STATE (2) estimate $\widehat\bw^k$ using a stationary point of $ \underset{\bw}{\text{minimize}} \ \frac{1}{2} \bw^\top \hat\bSigma^{k}_1\bw - \bar{\bx}_1^{k\top} \bw + P_{\lambda}(\bw)$
\STATE (3) project $\calD_2^{\pi_k}$ and obtain $y_i^k = \hat\bw^{k\top} \bx_{\pi_{k}(i)}, i=n_1+1,\dots, n$
\STATE (4) $T_{\hat\bw^k}= \sqrt{n_2}\widebar{y}^k/s^k_y$, where $\widebar{y}^k$ and $(s^k_y)^2$ are the sample mean and variance of $\{ y_{n_1+1}^k, \dots, y_n^k \}$
\STATE (5) compute the $p$-values by $p_k =  2\left( 1-\Phi(|T_{\widehat\bw^k}|) \right) $
\ENDFOR
\STATE \textbf{Step 3:} combine the $p$-values
\STATE \quad\ (1) compute the sample mean $\widebar{Z}$ and variance $s_{Z}^2$  of $\{Z_k = \Phi^{-1}(p_k), k=1,\dots, m\}$
\STATE \quad\ (2) compute test statistic $M_{\widehat\rho} = {\widebar Z}/{\sqrt{(1 + (m-1)\widehat\rho)/m}}$ 
\STATE \textbf{Return:} Reject $H_0$ at level $\alpha$ if $|M_{\widehat\rho}| > c(m, \alpha/2)$
\end{algorithmic}
\end{algorithm}


Two estimators of $\rho$ are introduced in \cite{follmann2012test}. Let $s_Z^2$ is the sample variance of $Z_i$'s. The first estimator is given by $\hat\rho_1 = \max(0, 1-s_Z^2)$. The second estimator, which is quantile based, is given by
$\hat\rho_{2} = \max(0, 1 - (m-1)s_Z^2/\chi^2_{m-1}(1-\beta)),$
where $\chi^2_{m-1}(1-\beta)$ is the upper $(1-\beta)$ quantile of $\chi^2_{m-1}$. An appealing approach is to choose $\beta$ as large as possible so that the test with $c(m, \alpha/2) = z_{\alpha/2}$ remains level $\alpha$ for all $\rho$. We refer to \cite{follmann2012test} for more details. \cite{follmann2012test} also provides the table of critical values $c(m, \alpha/2)$ and $\beta$ for different $m$, see Tables \ref{tab:Mrho1} and \ref{tab:Mrho2} in Appendix \ref{supp:critical_values}. We see that $c(m, \alpha/2)$ increases dramatically as $m$ increases, leading to low power when $m$ is large. Hence the quantile approach is preferred when $m$ is relatively large.


As for the choice of $m$, Figure \ref{fig:n40_multi_m} shows how the testing power
changes with the number of splits $m$ 
under settings with different correlation structure, samples sizes and dimensions.
We would like to point out that the proposed MPT is a valid level $\alpha$ test regardless the dependence structure
since the critical value is chosen such that Type I error is controlled for all dependence structure (i.e., for all $\rho$).
In other words, the proposed MPT is able to control Type I error regardless the number of splits $m$ and data
characteristics (e.g., sample size, data dimensionality). The main purpose of conducting multiple splitting
(i.e., choosing $m > 1$) is to mitigate power loss brought by single data-splitting procedure and
improve the testing power over SPT. Under alternative hypothesis, the testing power from each
split depends on the original data characteristics. Hence  theoretically the exact relationship
between $m$ and the testing power of MPT also depend on original data characteristics.
As shown in the plot, the testing power increases as $m$ increases but the improvement of
power becomes insignificant when $m$ is relatively large. 
Considering the fact that a large number of splits will increase computational cost, we recommend to set $m \in [30, 60]$ as a reasonable choice in practice considering the trade-off between testing power and computational cost.

\begin{figure}[!h]
\centering
    \includegraphics[width=0.9\textwidth]{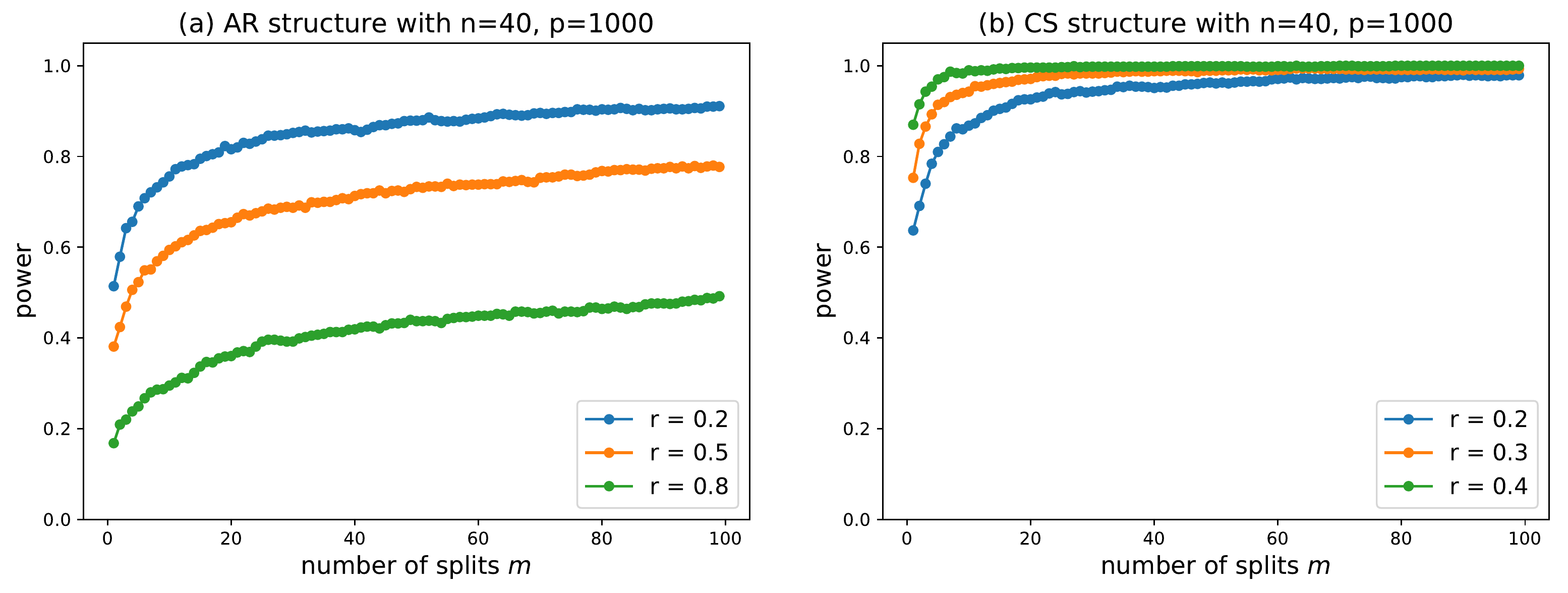}
    \includegraphics[width=0.9\textwidth]{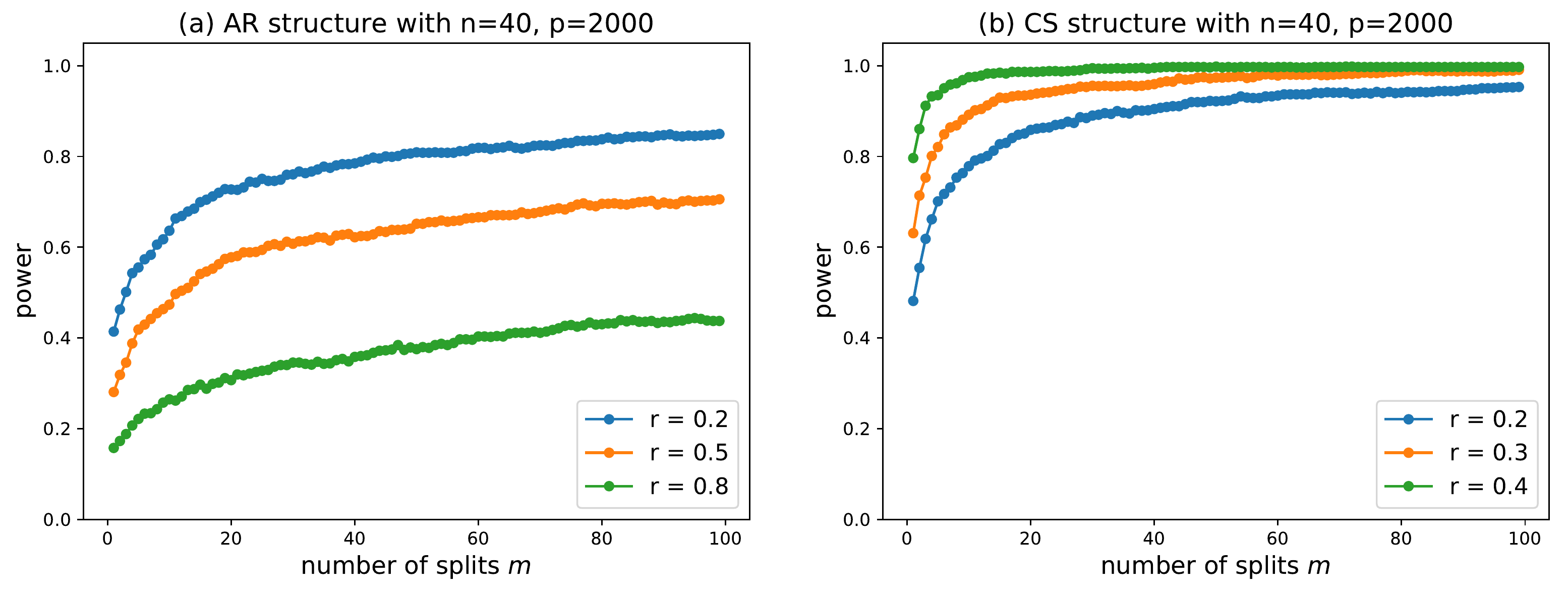}
\caption{Testing power versus number of splits $m$ for autocorrelation (AR) structure and compound symmetry (CS) structure with different choices of $n$ and $p$. Different colors correspond to different correlations $r$.}
\label{fig:n40_multi_m}
\end{figure}

\section{Numerical Studies}
\label{sec:numerical}
In this section, we conduct numerical studies to demonstrate the finite-sample performance of the proposed MPT through both Monte Carlo simulation and a real data example. 
\subsection{Monte Carlo Simulation}
\label{subsec:simulation}

We compare the proposed MPT with other state-of-the-art tests and $p$-value combination approaches.
In particular, we include the following tests in our experiments:

\begin{itemize}
    \item \textbf{Projection test}: our proposed SPT and MPT (with $m=40$), ridge projection test (Ridge) \citep{li2015projection}, and random projection test (RPT) \citep{lopes2011more}.
    \item \textbf{Combining $p$-values}: Median-based combination (Median) \citep{romano2019multiple}, average-based combination (Average)  \citep{ruschendorf1982random}, average-based combination using $\Phi^{-1}(p_k)$ (Z-average)
    \citep{romano2019multiple}, and Cauchy combination (Cauchy) \citep{liu2020cauchy}.
    \item \textbf{Quadratic-form test}: CQ test \citep{chen2010two}.
    \item \textbf{Extreme-type test}: CLX test \citep{cai2014two}. 
\end{itemize}

We generate a random sample of size $n$ from $N_p(c\bmu, \bSigma)$ with $\bmu = (\bfm 1^\top_{10}, \bfm 0^\top_{p - 10})^\top$. We set $c = 0, 0.5$ to examine the size and the power of these tests, respectively.
To examine the test robustness to non-normally distributed data, we also generate random samples from a multivariate $t_6$-distribution.
Let $\sigma_{ij}$ be the $(i,j)$ entry in $\bSigma$. For $r \in (0,1)$, we consider the following two covariance matrices: (1) compound symmetry (CS) with $\sigma_{ij} = r$ if $i \neq j$ and  $\sigma_{ij} = 1$ if $i=j$ and (2)  autocorrelation (AR) with $\sigma_{ij} = r^{|i-j|}$.
We vary $r$ from $0.1$ to $0.9$ with step size $0.1$ to examine the impact of correlation on size and power. We set sample size $n = 40, 100$ and dimension $p = 1000$.

In the above settings, the optimal projection direction $\bSigma^{-1}\bmu$ is sparse or approximately sparse.
When $\bSigma$ has the compound symmetry structure, $\bSigma^{-1}$ is an approximately sparse matrix in the sense that the off-diagonal entries are of order $p^{-1}$ and dominated by its diagonal entries. Then optimal projection direction $\bSigma^{-1}\bmu$ is also approximately sparse since the first $10$ entries dominate the rest entries. When $\bSigma$ has the autocorrelation structure, $\bSigma^{-1}$ is a 3-sparse matrix, meaning that at most three entries in each row or column are nonzero, and the resulting optimal projection direction $\bSigma^{-1}\bmu$ is sparse as well.
We set $\kappa=0.5$ when implementing the SPT and the MPT, i.e., half of the sample is used to estimate the projection direction and the other half is used to perform the test. The quantile approach $\hat\rho_2$ is used to estimate pairwise correlation among $Z_k$'s.
We set the type I error rate $\alpha=0.05$. All simulation results are based on 10,000 independent replications. 

\begin{figure}[!htb]
\centering
    \includegraphics[width=0.95\textwidth]{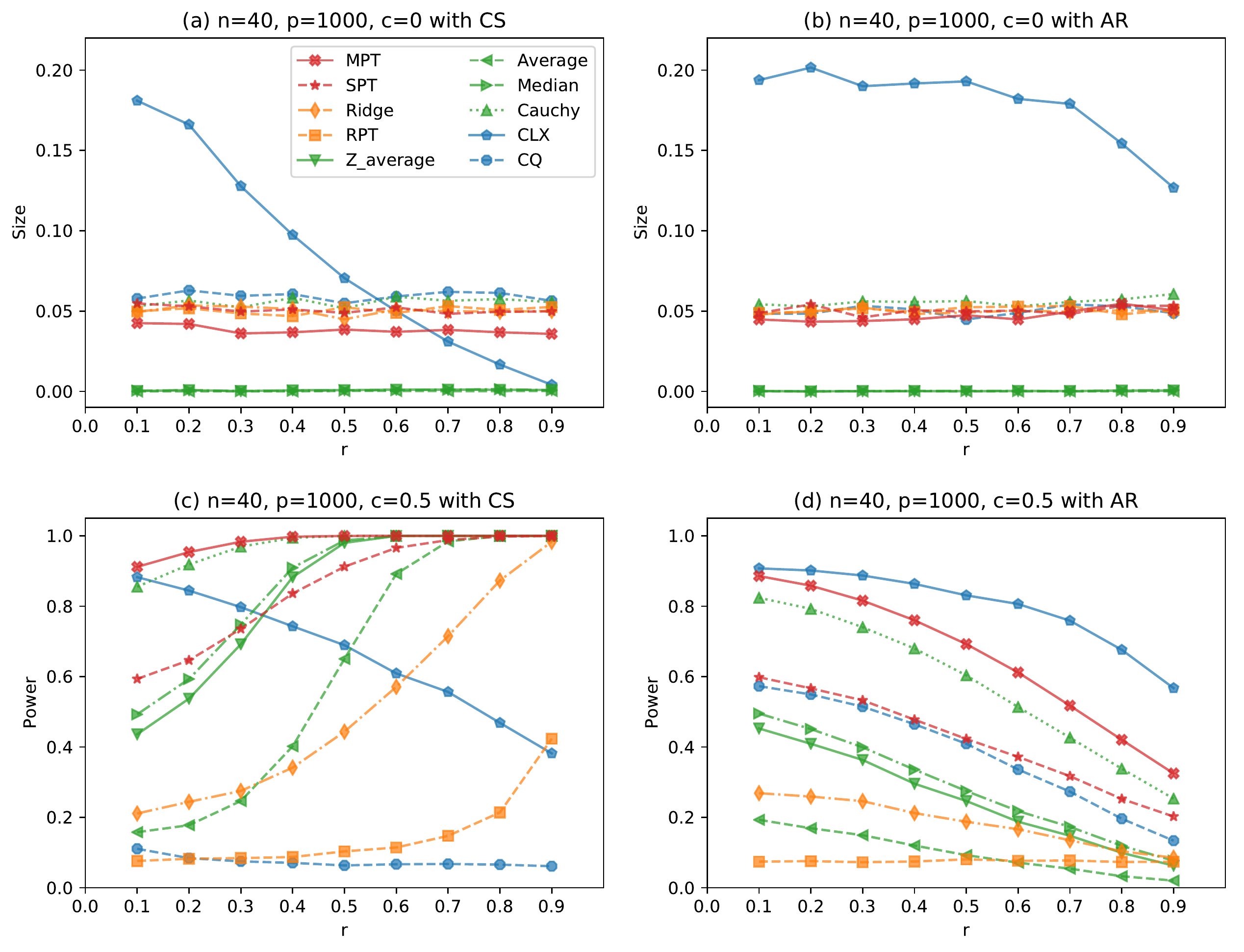}
\caption{Size and power of different tests for normally distributed samples with $n=40$. Panels (a) and (b) show size ($c=0$) under the null hypothesis for the CS and AR structure, respectively. Panels (c) and (d) show power ($c=0.5$) under the alternative hypothesis for the CS and AR structure, respectively.}\label{fig:n40_normal}
\end{figure}

Figure \ref{fig:n40_normal} reports the size and power of different tests for normally distributed samples with $n=40$. 
In terms of size, the proposed MPT successfully controls the type I error rate below $\alpha$. It is slightly conservative since the critical value is chosen against the worst $\rho$. 
The size of Cauchy test and CQ test are slightly inflated. The Median test, Average test and Z-average test are too conservative and their size are very close to $0$.
The CLX test completely fails to control the type I error rate 
due to the slow convergence rate of the limiting distribution.
As for power analysis, under the CS structure, the MPT outperforms all other tests. 
Cauchy test is slightly less powerful than the MPT but more powerful than other conservative combination approaches. 
The power of CLX test and CQ test decreases as the correlation increases since both tests ignore the dependence among variables.
In addition, CLX test and CQ test require the largest eigenvalue of $\bSigma$ is upper bounded by some constant, which is not satisfied in the CS structure.
Under the AR structure, 
note that the CLX test cannot control the size at all under $H_0$ (can be as large as $0.20$). 
The size inflation makes the testing power artificially high and hence it is not trustworthy.
Excluding the CLX test, the proposed MPT has the best performance, followed by the Cauchy test. 
Other conservative combination tests are even less powerful than the SPT, which indicates such conservative combination methods do not necessarily improve the testing power.

\begin{figure}[!htb]
\centering
    \includegraphics[width=0.95\textwidth]{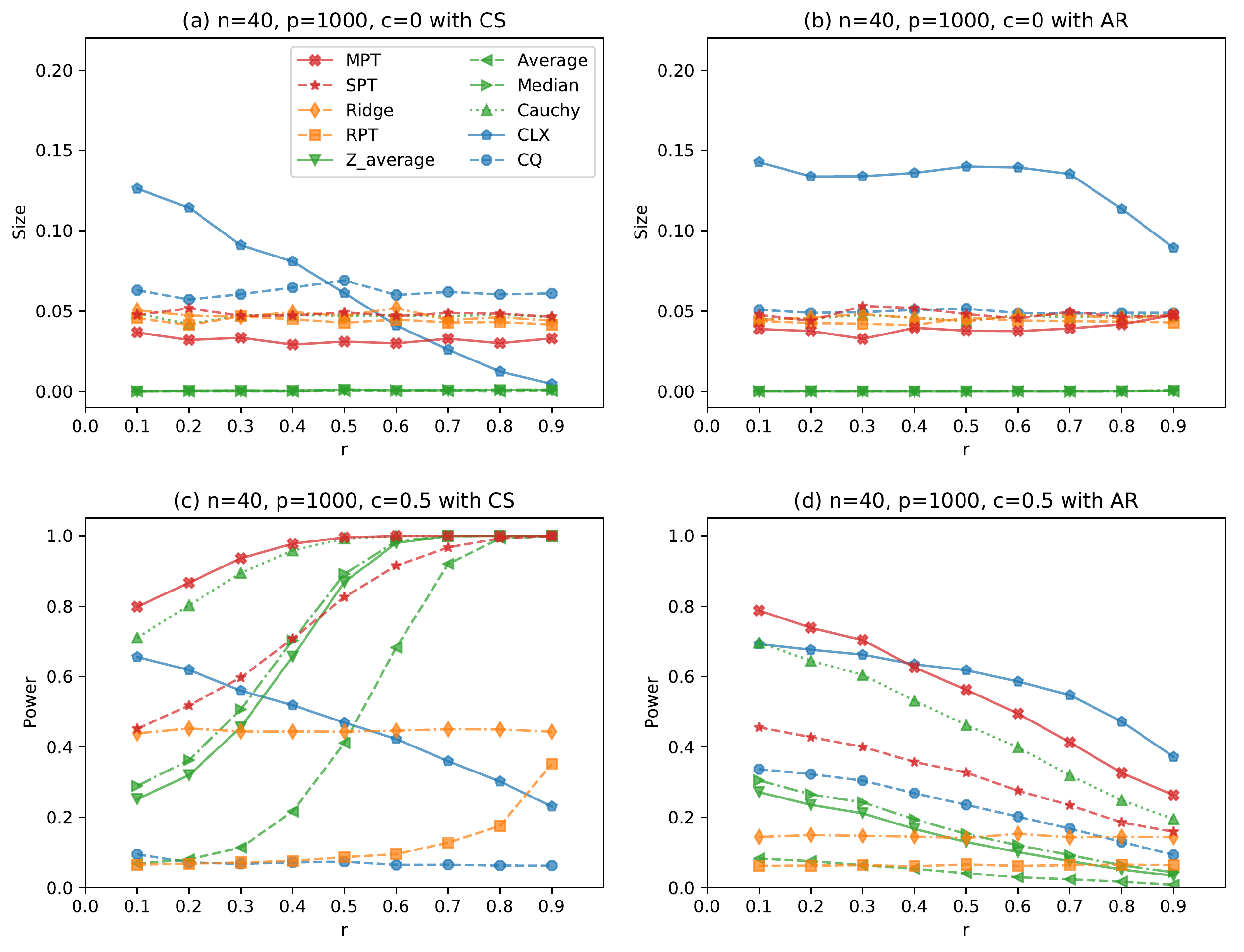}
\caption{Size and power of different tests for samples following multivariate $t_6$ distribution with $n=40$. Panels (a) and (b) show size ($c=0$) under the null hypothesis for the CS and AR structures, respectively. Panels (c) and (d) show power ($c=0.5$) under the alternative hypothesis for the CS and AR structures, respectively.}
\label{fig:n40_t6}
\end{figure}

We also examine the finite-sample performance of the proposed MPT when the normality assumption is not satisfied.
Figure \ref{fig:n40_t6} shows the size and power comparisons of different tests when $\bx_i$'s are generated from multivariate $t_6$ distribution with AR and CS covariance structure. The results show a similar pattern as those in the normal settings, which provide numerical evidences on the robustness of the MPT to non-Gaussianity. When $n=100$, the patterns of size and power are similar to that of $n=40$. Due to the limit of space, we relegate the results for $n=100$ to the appendix, see Figures \ref{fig:n100_normal} and \ref{fig:n100_t6} in Appendix \ref{supp:numeric}.

The numerical results in this subsection emphasize that the MPT greatly improves the testing power upon the SPT thanks to the multiple splits. In a brief summary, our proposed MPT successfully controls the type I error rate and achieves the highest testing power in comparison with all other state-of-the-arts level $\alpha$ tests. In addition, the studies reveal that the proposed MPT is quite robust to non-Gaussianity.

\subsection{Real Data Analysis}
\label{subsec:realdata}
We apply the proposed MPT and SPT together with other tests introduced above to a real dataset of high resolution micro-computed tomography \citep{percival2014embryonic}. This dataset contains skull bone densities of $n=29$ mice with genotype ``T0A1'' in a genetic mutation study. For each mouse, bone density is measured for 16 different areas of its skull at density levels between 130 - 249. 
In this empirical analysis, we are interested in comparing the bone density patterns of two areas in the skull, namely ``Mandible'' and ``Nasal''. We use all density levels between 130 - 249 for our analysis, and hence dimension $p = 120$.
Since the two areas come from the same mouse, we first take the difference of bone density in the two areas at the corresponding density level for each observation. Then we normalize the bone density in the sense that $\frac{1}{29}\sum_{i=1}^{29} X^2_{ij} = 1$ for all $1 \leq j \leq 120$. The null hypothesis is the density patterns of two skull areas are the same.

\begin{table}[h]
\caption{Decisions on whether null hypothesis should be rejected or not at significance level $\alpha=0.05$ based on different tests for the bone density dataset with various signal strengths. The numbers in the parentheses in the $p$-values if applicable.}
\centering
\resizebox{1.\columnwidth}{!}{
\begin{tabular}{lllllllllllllllll}
\toprule
$\delta$ &  1.0 & 0.8 & 0.6 & 0.4 & 0.3 & 0.2 & 0.18 \\
\midrule
MPT & \cmark & \cmark & \cmark & \cmark & \cmark & \cmark & \cmark\\
Cauchy & \cmark & \cmark & \cmark & \cmark & \cmark & \cmark & \cmark\\
Median & \cmark & \cmark & \cmark & \cmark & \cmark & \xmark & \xmark\\
Average & \cmark & \cmark & \cmark & \cmark & \xmark & \xmark & \xmark\\
Z-average & \cmark & \cmark & \cmark & \cmark & \cmark & \xmark & \xmark\\
\midrule

SPT & \cmark~$(10^{-10})$ & \cmark~$(10^{-9})$ & \cmark~$(10^{-7})$ & \cmark~$(10^{-6})$ & \cmark~$(10^{-4})$ & \cmark~(0.042) & \xmark~(0.246)  \\
Ridge & \cmark~$(10^{-8})$ & \cmark~$(10^{-7})$ & \cmark~$(10^{-5})$ & \cmark~(0.001) & \cmark~(0.014) & \xmark~(0.146) & \xmark~(0.387) \\
RPT & \cmark~$(10^{-9})$ & \cmark~$(10^{-8})$ & \cmark~$(10^{-6})$ & \cmark~$(10^{-4})$ & \cmark~$(0.010)$ & \xmark~(0.203) & \xmark~(0.347) \\
CQ & \cmark~(0) & \cmark~(0) & \cmark~(0) & \cmark~$(10^{-4})$ & \xmark~(0.081) & \xmark~(0.772) & \xmark~(0.945) \\
CLX & \cmark~(0) & \cmark~$(10^{-14})$ & \cmark~$(10^{-8})$ & \cmark~$(0.004)$ & \xmark~(0.189) & \xmark~(0.965) & \xmark~(0.994) \\
\bottomrule
\end{tabular}
}
\label{pvalue}
\end{table}

We apply the proposed MPT and SPT together with other tests to the bone density dataset.
The decisions as well as $p$-values if applicable (in the parentheses) are reported in the first column in Table \ref{pvalue}. All tests are able to reject the null hypothesis, implying that the bone density patterns are significantly different. To compare the power of different tests, we further conduct tests as we decrease the signal strength in the bone density difference. To be specific, let $\bar\bx$ be the sample mean and $\br_i = \bx_i - \bar\bx$ be the residual for the $i$th subject. Then a new observation $\bz_i = \delta\bar\bx + \br_i$ is constructed for the $i$th subject for some $\delta \in [0,1]$. By this construction, a smaller $\delta$ leads to a weaker signal strength and would make the test more challenging. Table \ref{pvalue} also reports the decisions and $p$-values (in the parentheses) for $\delta = 1.0, 0.8, 0.6, 0.4, 0.3, 0.2, 0.18$.
When $\delta \geq 0.4$, all the tests perform well and reject the null hypothesis at level $0.05$. When $\delta$ decreases to $0.3$, the CQ test, the CLX test and the average based combination test start to fail to reject the null hypothesis.
When $\delta = 0.2$, the proposed MPT and SPT and Cauchy combination test are able to reject the null. Further decreasing $\delta$ to $0.18$, only the MPT and Cauchy combination test can detect the density difference.
This real data example demonstrates that the proposed MPT is more powerful than existing tests and performs well even when the signal is very weak.


\section{Discussion}
\label{sec:discuss}
In this work, we study the hypothesis test for one-sample mean vectors in high dimensions. We first study the question of estimating optimal projection direction and provide statistical guarantee. Furthermore, we propose the multiple-splitting projection test, which makes use of the exchangeability of multiple $p$-values, to mitigate the power loss arising from the single data-splitting procedure.
The proposed multiple data-splitting framework can be easily extended to a two-sample problem in which the optimal projection direction is $\bSigma^{-1}(\bmu_1 - \bmu_2)$ \citep{li2015projection}. Sharing the same spirit, half of the sample can be used to estimate the projection direction and the remaining half is used to perform the two-sample $t$-test. Then resultant multiple $p$-values can be combined similarly to the MPT.
As pointed out in Theorem \ref{thm:general_exchangeable}, the exchangeability phenomenon generally holds for a permutation framework. This work can be extended to many other statistical inference problems, such as testing the coefficients for a high-dimensional regression model. We hope such insight provides new ideas for researchers from related areas.  Another interesting extension is to develop more refined combination methods which better handle the exchangeability. We leave these interesting questions as future work.

\bibliographystyle{agsm}
\bibliography{ref-mpt}

\newpage
\section*{Appendices}
\appendix
The appendices provide additional materials for the main manuscript. Appendix \ref{supp:critical_values} provides the tables of critical values for the proposed MPT.
Appendix \ref{supp:proof} presents  technical lemmas and complete proofs of theoretical results. Appendix \ref{supp:numeric} reports additional numerical results of size and power comparisons for $n=100$ to serve as a complementary to the numerical studies in Section \ref{sec:numerical}.

\section{Tables of Critical Values}
\label{supp:critical_values}

In this section, we present the tables of critical values for the proposed MPT.
\cite{follmann2012test} derives the critical values of $c(m, \alpha/2)$ and $\beta$ for tests $M_{\hat\rho_1}$ and $M_{\hat\rho_2}$ at level $\alpha=0.05$, respectively. We summarize the critical values in Tables \ref{tab:Mrho1}-\ref{tab:Mrho2}.

\begin{table}[H]
\caption{Critical value $c(m, \alpha/2)$ with respect to $m$ for the test $M_{\widehat\rho_1}$ at level $\alpha=0.05$ }
\small
\centering
\begin{tabular}{ccccccccccccccc}
\toprule
& & \multicolumn{10}{c}{number of splits $m$} \\
\cmidrule(r){3-12}
Method & Value & 2 & 3 & 4 & 5 & 10 & 20 & 40 & 100 & 1000 & 10000 \\
\midrule
$M_{\widehat\rho_1}$ & $c(m,\alpha/2)$ & 1.988 & 2.058 & 2.133 & 2.204 & 2.489 & 2.865 & 3.126 & 4.115 & 7.17 & 12.66 \\
\bottomrule
\end{tabular}
\label{tab:Mrho1}
\end{table}

\begin{table}[H]
\caption{\protect\centering The smallest value $\beta$ with respect to $m$ to control the type I error of the test $M_{\widehat\rho_2}$ \newline with $c(m, \alpha/2) = z_{\alpha/2}$ at level $\alpha=0.05$}
\small
\centering
\begin{tabular}{ccccccccccccccc}
\toprule
& & \multicolumn{10}{c}{number of splits $m$} \\
\cmidrule(r){3-12}
Method & Value & 2 & 3 & 4 & 5 & 10 & 20 & 40 & 100 & 1000 & 10000 \\
\midrule
$M_{\widehat\rho_2}$ & $\beta$ & 0.25 & 0.25 & 0.25 & 0.25 & 0.20 & 0.20 & 0.15 & 0.15 & 0.10 & 0.05 \\
\bottomrule
\end{tabular}
\label{tab:Mrho2}
\end{table}




\section{Lemmas and Proofs}
\label{supp:proof}

\subsection{Technical Lemmas}
\label{supp:lemmas}
In this subsection, we introduce a few technical lemmas to help establish theoretical results. Before proceeding, we first introduce some notations for sub-Gaussian and sub-exponential random variables. The sub-Gaussian norm of a sub-Gaussian random variable $X$ is
\begin{equation*}
   \|X\|_{\psi_2}= \sup_{p\geq 1} p^{-\frac{1}{2}}(\E|X|^p)^{1/p}.
\end{equation*}
The sub-exponential norm of a sub-exponential random variable $X$ is
\begin{equation*}
   \|X\|_{\psi_1}= \sup_{p\geq 1} p^{-1}(\E|X|^p)^{1/p}.
\end{equation*}

\begin{lemma}
    If the RSC condition (\ref{eqn:rsc}) holds, then
    $$
    \bDelta^\top\bW\bDelta \geq \nu \|\bDelta\|_2^2 - \tau\sqrt{\frac{\log p}{n}} \|\bDelta\|_1 \ \text{for all} \ \bDelta\in\RR^p.
    $$
\label{lem:full_rsc}
\end{lemma}

\begin{proof}
For any $\|\bDelta\|_1 < 1$, the $L_1$ norm of $\bDelta/\|\bDelta\|_1$ is 1 and hence satisfies the RSC condition in (\ref{eqn:rsc}). We have
\begin{align*}
\begin{aligned}
    \frac{\bDelta^\top}{\|\bDelta\|_1}\bW\frac{\bDelta}{\|\bDelta\|_1} &\geq
    \nu\frac{\|\bDelta\|_2^2}{\|\bDelta\|_1^2} - \tau\sqrt{\frac{\log p}{n}} \frac{\|\bDelta\|_1}{\|\bDelta\|_1} \\
    \frac{\bDelta^\top}{\|\bDelta\|_1}\bW\frac{\bDelta}{\|\bDelta\|_1} &\geq
    \nu\frac{\|\bDelta\|_2^2}{\|\bDelta\|_1^2} - \tau\sqrt{\frac{\log p}{n}} \frac{\|\bDelta\|^2_1}{\|\bDelta\|^2_1} \\
    \bDelta^\top\bW\bDelta &\geq \nu \|\bDelta\|_2^2 - \tau\sqrt{\frac{\log p}{n}} \|\bDelta\|_1^2. \\
\end{aligned}
\end{align*}
Since $\|\bDelta\|_1<1$, then $\|\bDelta\|_1^2 \leq \|\bDelta\|_1$, implying
\begin{equation*}
    \bDelta^\top\hat\bSigma\bDelta \geq \nu \|\bDelta\|_2^2 - \tau\sqrt{\frac{\log p}{n}} \|\bDelta\|_1.
\end{equation*}
The proof of Lemma \ref{lem:full_rsc} is complete.
\end{proof}

\begin{lemma}
Suppose $\bx_1, \dots, \bx_n \in \RR^p \sim (\bmu, \bSigma)$ are independent and identically distributed sub-Gaussian random vectors. Let $\bar\bx$ and $\hat{\bSigma} = (\hat \sigma_{ij})_{p\times p}$ be the sample mean and sample covariance matrix. If $\log p < n$, then with probability at least $1 - 2p^{-1}$, we have
\begin{enumerate}[(i)]
    \item $\|\bar\bx - \bmu\|_\infty \leq C\sqrt{\log p/n}$ for some large $C$.
    \item $\|\hat{\bSigma} - \bSigma\|_{\max} \leq C\sqrt{\log p/n}$ for some large $C$.
\end{enumerate}
\label{lem:sam_cov_rate}
\end{lemma}

\begin{proof}
Let $\bar\bx=\frac{1}{n}\sum_{k=1}^n \bx_k$ be the sample mean and $\hat{\bSigma} = \frac{1}{n}\sum_{k=1}^n(\bx_k - \bar \bx)(\bx_k - \bar \bx)^\top$ be the sample covariance matrix.
Without loss of generality, we assume $\E(\bx_i) = \bfm 0$ and the sub-Gaussian parameter for $\bx_i$ is $\sigma^2$. Write $\bx_k = (\bx_{k1}, \dots, \bx_{kp})^\top$ and each $X_{kj}$ is a sub-Gaussian random variable with parameter $\bsigma^2$ and let $K = \max_{1\leq j \leq p}  \|\bx_{kj}\|_{\psi_2}$.
Obviously, $\bar\bx$ is also sub-Gaussian random vector with parameter $\sigma^2/\sqrt{n}$. For any $t>0$, we have
\begin{equation*}
    \prob(\| \bar\bx - \bmu \|_\infty > t) \leq 2p\exp\left\{-{cnt^2}/{K^2}\right\}.
\end{equation*}
Take $t = C\sqrt{\log p/n}$ for some large $C>0$, we have
\begin{equation}
    \prob(\| \bar\bx - \bmu \|_\infty < C\sqrt{\log p/n}) \geq 1 - 2p^{-1}.
\label{eqn:bound_mu}
\end{equation}
The sample covariance $\hat\bSigma$ can be decomposed as
\begin{equation*}
  \hat{\bSigma} = \frac{1}{n}\sum_{k=1}^n \bx_k\bx_k^\top - \bar{\bx}\bar{\bx}^\top.
\end{equation*}
Hence we know,
$$
  \max_{i,j}|\hat{\sigma}_{ij} - \sigma_{ij}| \leq \max_{i,j}|\frac{1}{n}\sum_{k=1}^n\bx_{ki}\bx_{kj} - \sigma_{ij}| + \max_{i,j}|\bar{\bx}_i\bar{\bx}_j|.
$$
In addition, we have
\begin{equation*}
  \|\bx_{ki}\bx_{kj}\|_{\psi_1} \leq 2 \|\bx_{ki}\|_{\psi_2} \|\bx_{kj}\|_{\psi_2} \leq 2K^2.
\end{equation*}
Hence $\|\bx_{ki}\bx_{kj} - \sigma_{ij}\|_{\psi_1} \leq 4K^2$.
According to the inequality of tail probability for sub-exponential variables, we have
\[
\prob\left(\max_{i,j} \left|\frac{1}{n}\sum_{k=1}^n\bx_{ki}\bx_{kj} - \sigma_{ij}\right| > t\right) \leq \max\left( 2p^2\exp\left\{-cn\frac{t^2}{16K^4}\right\}, 2p^2\exp\left\{-cn\frac{t}{4K^2}\right\}\right).
\]
It is easy to verify that $\|\bar{\bx}_i\|_{\psi_2} \leq K/\sqrt{1/n}$ and $\|\bar{\bx}_i\bar{\bx}_j\|_{\psi_1} \leq 2\|\bar{\bx}_i\|_{\psi_2}\|\bar{\bx}_j\|_{\psi_2} \leq 2K^2/n$,
we have
\[\prob(\max_{i,j}|\bar{\bx}_i\bar{\bx}_j| > t) \leq 2p^2\exp\left\{-\frac{cnt}{2K^2}\right\}.
\]
By the choice of $t = \frac{C}{2}\sqrt{\frac{\log p}{n}}$ for some large $C>0$, we have $\max_{i,j}|\hat{\sigma}_{ij} - \sigma_{ij}| \leq C\sqrt{\log p /n}$ with probability at least $1-2p^{-1}$, which completes the proof.

\end{proof}

\begin{lemma}[\cite{loh2015regularized}]
  Assume penalty function $P_\lambda(t)$ satisfies conditions (i)-(v), then
  \begin{enumerate}[(a)]
    \item $|P_\lambda(t_1) - P_\lambda(t_2)| \leq \lambda|t_1 - t_2|$ for any $t_1, t_2 \in \RR$.
    \item For any $\bw \in \RR^p$, we have $\lambda\|\bw\|_1 \leq P_\lambda(\bw) + \frac{\nu}{2}\|\bw\|_2^2$.
    \item Suppose $\|\bw^\star\|_0 = s > 0$, then for any $\bw \in \RR^p$ such that $c P_\lambda(\bw^\star) - P_\lambda(\bw) \geq 0$ with $c \geq 1$, we have $c P_\lambda(\bw^\star) - P_\lambda(\bw) \leq \lambda(c\|\bdelta_\calA\|_1 - \|\bdelta_{\calA^c}\|_1)$, where $\bdelta = \bw - \bw^\star$ and $\calA$ is the index set of the $s$ largest elements of $\bdelta$ in magnitude.
    \item Define $J_\lambda(t) = \lambda|t| - P_\lambda(t)$. Then the function $J_\lambda(t) - \frac{\mu}{2}t^2 = \lambda|t| - P_\lambda(t) - \frac{\mu}{2}t^2$ is concave and differentiable.
  \end{enumerate}
\label{lem:penalty}
\end{lemma}

\subsection{Proof of Theorem \ref{thm:error_bound}}
Lemma \ref{lem:full_rsc} shows that the RSC condition in (\ref{eqn:rsc}) actually holds for all $\bDelta \in \RR^p$. Now we are ready to prove Theorem~\ref{thm:error_bound}.
Define $\bw^\star = \bSigma^{-1}\bmu$ and $\hat\bDelta = \hat\bw - \bw^\star$. The first order necessary condition (\ref{con:first_order}) implies that
\begin{equation}
    \hat\bDelta^\top\hat\bSigma\hat\bw + \langle \nabla P_\lambda(\hat\bw) - \bar\bx, \hat\bDelta \rangle = 0.
\label{eqn:1storder1}
\end{equation}
By the RSC condition (\ref{eqn:rsc}), we have
\begin{equation}
    \hat\bDelta^\top\hat\bSigma\hat\bDelta \geq \nu\|\hat\bDelta\|_2^2 - \tau\sqrt{\frac{\log p}{n}} \|\hat\bDelta\|_1.
\label{eqn:rsc1}
\end{equation}
Adding (\ref{eqn:1storder1}) to (\ref{eqn:rsc1}), we have
\begin{equation}
    -\hat\bDelta^\top\hat\bSigma\bw^\star - \langle \nabla P_\lambda(\hat\bw) - \bar\bx, \hat\bDelta \rangle \geq \nu\|\hat\bDelta\|_2^2 - \tau\sqrt{\frac{\log p}{n}} \|\hat\bDelta\|_1.
\label{eqn:myineq1}
\end{equation}
Lemma \ref{lem:penalty} shows that $P_{\lambda,\gamma}(\bw) = P_{\lambda}(\bw) + \frac{\gamma}{2}\|\bw\|_2^2$ is a convex function, hence
\begin{equation*}
    P_{\lambda,\gamma}(\bw^\star) - P_{\lambda,\gamma}(\hat\bw) \geq \langle \nabla P_\lambda(\hat\bw) + \gamma \hat\bw,  \bw^\star - \hat\bw\rangle,
\end{equation*}
which implies
\begin{equation}
    - \langle \nabla P_\lambda(\hat\bw),  \hat\bDelta \rangle \leq P_\lambda(\bw^\star) - P_\lambda(\hat\bw) + \frac{\gamma}{2}\| \hat\bDelta \|_2^2.
\label{eqn:bound1}
\end{equation}
Combining  (\ref{eqn:myineq1}) and (\ref{eqn:bound1}),
\begin{equation*}
\begin{aligned}
    \nu\|\hat\bDelta\|_2^2 - \tau\sqrt{\frac{\log p}{n}} \|\hat\bDelta\|_1
    &\leq -\hat\bDelta^\top\hat\bSigma\bw^\star + \bar\bx^\top\hat\bDelta + P_\lambda(\bw^\star) - P_\lambda(\hat\bw) + \frac{\gamma}{2}\| \hat\bDelta \|_2^2 \\
    (\nu - {\gamma}/{2})\| \hat\bDelta \|_2^2 &\leq
    P_\lambda(\bw^\star) - P_\lambda(\hat\bw) + \|\tilde\bSigma\bw^\star - \bar\bx\|_\infty \|\hat\bDelta\|_1 + \tau\sqrt{\frac{\log p}{n}} \|\hat\bDelta\|_1 \\
    (\nu - {\gamma}/{2})\| \hat\bDelta \|_2^2
    &\leq P_\lambda(\bw^\star) - P_\lambda(\hat\bw) + \left( \|\hat\bSigma\bw^\star - \bar\bx\|_\infty + \tau\sqrt{\frac{\log p}{n}} \right) \|\hat\bDelta\|_1. \\
\end{aligned}
\end{equation*}
By triangle inequality, we know
$ \|\hat{\bSigma}\bw^\star - \bar\bx\|_\infty \leq \|\hat{\bSigma}\bw^\star - \bmu\|_\infty + \| \bar\bx - \bmu \|_\infty.$
According to Lemma \ref{lem:sam_cov_rate}, there exists $M_1, M_2>0$ such that
\begin{equation}
    \prob(\| \bar\bx - \bmu \|_\infty \leq M_1\sqrt{\log p/n}) \geq 1 - 2p^{-1}.
\label{eqn:bound_mu}
\end{equation}
\begin{equation*}
    \prob(\|\hat\bSigma - \bSigma\|_{\max} \leq M_2\sqrt{\log p/n} ) \geq 1 - 2p^{-1}.
\end{equation*}
Then with probability at least $1-2p^{-1}$,
\begin{equation}
\begin{aligned}
    \|\hat{\bSigma}{\bw^\star} - \bmu\|_\infty &= \|\hat{\bSigma}{\bw^\star} - \bSigma{\bw^\star} \|_\infty
    \leq \|\hat\bSigma - \bSigma\|_\infty \|\bw^\star\|_1 
    \leq M_2C_1\sqrt{\log p/n}.
\end{aligned}
\label{eqn:bound_cov}
\end{equation}

Combining (\ref{eqn:bound_mu}) and (\ref{eqn:bound_cov}), we know that with probability at least $1-4p^{-1}$, we have
\begin{equation*}
    \|\hat{\bSigma}{\bw^\star} - \bar\bx\|_\infty \leq M'\sqrt{\log p/n},
\end{equation*}
with $M' = M_1 + M_2C_1$.
Take $\lambda = M\sqrt{\log p/n}$ with $M = 4\max\{ M', \tau \}$, we have
$$
    \|\hat\bSigma\bw^\star - \bar\bx\|_\infty + \tau\sqrt{{\log p}/{n}} \leq \lambda/2.
$$
Hence
\begin{equation*}
\begin{aligned}
    (\nu - \gamma/2)\|\hat\Delta\|_2^2
    &\leq P_\lambda(\bw^\star) - P_\lambda(\hat\bw) + \frac{\lambda}{2} \|\hat\Delta\|_1 \\
    &\leq P_\lambda(\bw^\star) - P_\lambda(\hat\bw) + \frac{1}{2} P_\lambda(\hat\Delta) + \frac{\gamma}{4} \|\hat\Delta\|_2^2, \\
\end{aligned}
\end{equation*}
where the second inequality is because $\frac{\lambda}{2}\|\hat\Delta\|_1 \leq \frac{1}{2}P_\lambda(\hat\Delta) + \frac{\gamma}{4}\|\hat\Delta\|_2^2$ by Lemma \ref{lem:penalty}(b).
By the subadditivity of $P_\lambda$, we have $P_\lambda(\hat\Delta) = P_\lambda(\hat\bw-\bw^\star) \leq P_\lambda(\hat\bw) + P_\lambda(\bw^\star) $. Then
\begin{equation*}
\begin{aligned}
    \left(\nu - {\gamma}/{2}\right)\| \hat\Delta \|_2^2
    &\leq P_\lambda(\bw^\star) - P_\lambda(\hat\bw) + \frac{1}{2} P_\lambda(\hat\bw) + \frac{1}{2}P_\lambda(\bw^\star) + \frac{\gamma}{4} \|\hat\Delta\|_2^2 \\
    \left(\nu - {3\gamma}/{4}\right)\| \hat\Delta \|_2^2
    &\leq \frac{3}{2} P_\lambda(\bw^\star) - \frac{1}{2} P_\lambda(\hat\bw) \\
    \left(2\nu - {3\gamma}/{2}\right)\| \hat\Delta \|_2^2
    &\leq 3 P_\lambda(\bw^\star) -  P_\lambda(\hat\bw). \\
\end{aligned}
\end{equation*}
By Lemma \ref{lem:penalty}(c), we have
$3 \lambda\|\hat\Delta_\calI\|_1 - \lambda\|\hat\Delta_{\calI^c}\|_1 \geq 3 P_\lambda(\bw^\star) - P_\lambda(\hat\bw) \geq 0, $
where $\calI$ denotes the index set of the $s$ largest elements of $\hat\Delta$ in magnitude. Since $\nu \geq 3\gamma/4$, we have
\begin{equation*}
    0 \leq \left(2\nu - {3\gamma}/{2}\right)\| \hat\Delta \|_2^2
    \leq 3 \lambda\|\hat\Delta_\calI\|_1 - \lambda\|\hat\Delta_{\calI^c}\|_1.
\end{equation*}
As a result, we have $\|\hat\Delta_{\calI^c}\|_1 \leq 3 \|\hat\Delta_\calI\|_1$ and
\begin{equation*}
    \left(2\nu - \frac{3}{2}\gamma\right)\| \hat\Delta \|_2^2 \leq 3 \lambda\|\hat\Delta_\calI\|_1 - \lambda\|\hat\Delta_{\calI^c}\|_1 \leq  3 \lambda\|\hat\Delta_\calI\|_1 \leq 3 \lambda\sqrt{s}\|\hat\Delta_\calI\|_2,
\end{equation*}
from which we conclude that
\[
\| \hat\Delta \|_2 \leq \frac{6 \lambda\sqrt{s}}{4 \nu - 3 \gamma} = O\left(\sqrt{\frac{s\log p}{n}}\right).
\]
The $\ell_1$ norm bound follows immediately from the $\ell_2$ norm bound
\[
\| \hat\Delta \|_1 = \| \hat\Delta_\calI \|_1 + \| \hat\Delta_{\calI^c} \|_1 \leq 4 \| \hat\Delta_\calI \|_1 \leq 4 \sqrt{s}\| \hat\Delta_\calI \|_2 \leq \frac{24 \lambda{s}}{4 \nu - 3 \gamma} = O\left(s\sqrt{\frac{\log p}{n}}\right).
\]

\subsection{Proof of Theorem \ref{thm:SPT_power}}

According to Theorem~\ref{thm:error_bound}, we know
$\|\hat\bw - \bw^\star\|_1 = O(s\sqrt{\log p/n_1}) = o(1)$ with high probability. Let $\bar\bx_2$ and $\hat\bSigma_2$ be the sample mean and sample covariance matrix based on $\calD_2 = \{\bx_{n_1+1}, \dots, \bx_{n}\}$. On one hand,

\begin{equation*}
\begin{aligned}
    |\hat\bw^\top\bSigma\hat\bw - \hat\bw^{\top}\bSigma\bw^\star| &= |\hat\bw^\top\bSigma(\hat\bw - \bw^\star)| \\
    &\leq \|\bSigma\|_{\max} \|\hat\bw\|_1 \|\hat\bw - \bw^\star\|_1 \\
    &\leq \|\bSigma\|_{\max} (\|\bw^\star\|_1 + \|\hat\bw - \bw^\star\|_1) \|\hat\bw - \bw^\star\|_1 \\
    &= O\left(s\sqrt{{\log p}/{n_1}}\right).
\end{aligned}
\end{equation*}
One the other hand,
\begin{equation*}
\begin{aligned}
    |\hat\bw^\top\bSigma\bw^\star - \bw^{\star\top}\bSigma\bw^\star| &= |(\hat\bw - \bw^\star)^\top \bSigma \bw^\star| \\
    &\leq \|\bSigma\|_{\max} \|\bw^\star\|_1 \|\hat\bw - \bw^\star\|_1 \\
    &= O\left(s\sqrt{{\log p}/{n_1}}\right).
\end{aligned}
\end{equation*}
Hence by triangle inequality,
\begin{equation}
\begin{aligned}
    |\hat\bw^\top\bSigma\hat\bw - \bw^{\star\top}\bSigma\bw^\star| &\leq
    |\hat\bw^\top\bSigma\hat\bw - \hat\bw^{\top}\bSigma\bw^\star| + |\hat\bw^{\top}\bSigma\bw^\star - \bw^{\star\top}\bSigma\bw^\star| \\
    &= O\left(s\sqrt{{\log p}/{n_1}}\right).
\end{aligned}
\label{eqn:pt2_1}
\end{equation}


Given $\hat\bw$, we know that $y_{n_1+1}, \dots, y_n$ are i.i.d. random variables with mean $\bmu^\top\hat\bw$ and variance $\hat\bw^\top\bSigma\hat\bw$. By central limit theorem and $\hat\bw^\top\bSigma\hat\bw - \bw^{\star\top}\bSigma\bw^\star = o(1)$, we know
\begin{equation*}
\begin{aligned}
    \frac{\sqrt{n_2}(\bar y - \bmu^\top\hat\bw)}{\sqrt{\bw^{\star\top}\bSigma\bw^\star}} 
    \overset{d}{\rightarrow} N(0,1).
\end{aligned}
\end{equation*}
The test statistic of the SPT is $T_{\hat\bw} = \sqrt{n_2}\bar y/s_y$ and we reject $H_0$ whenever $|T_{\hat\bw}| > z_{\alpha/2}$. The power function for the SPT is
\begin{equation*}
\begin{aligned}
     \prob\left( \left|\frac{\sqrt{n_2}\bar y}{\sqrt{\bw^{\star\top}\bSigma\bw^\star}} \right| > z_{\alpha/2} \right)
    =& \prob\left( \frac{\sqrt{n_2}\bar y}{\sqrt{\bw^{\star\top}\bSigma\bw^\star}} > z_{\alpha/2} \right) + \prob\left( \frac{\sqrt{n_2}\bar y}{\sqrt{\bw^{\star\top}\bSigma\bw^\star}} <- z_{\alpha/2} \right) \\
    =& \prob\left( \frac{\sqrt{n_2}(\bar y - \bmu^\top\hat\bw )}{\sqrt{\bw^{\star\top}\bSigma\bw^\star}} > z_{\alpha/2} + \frac{\sqrt{n_2}\bmu^\top\hat\bw}{\sqrt{\bw^{\star\top}\bSigma\bw^\star}} \right) + \\
     & \prob\left( \frac{\sqrt{n_2}(\bar y -  \bmu^\top\hat\bw)}{\sqrt{\bw^{\star\top}\bSigma\bw^\star}} < -z_{\alpha/2} + \frac{\sqrt{n_2}\bmu^\top\hat\bw}{\sqrt{\bw^{\star\top}\bSigma\bw^\star}} \right) \\
    \simeq& \Phi\left( -z_{\alpha/2} + \frac{\sqrt{n_2}\bmu^\top\hat\bw}{\sqrt{\bw^{\star\top}\bSigma\bw^\star}} \right) + o(1) \\
    \simeq& \Phi\left( -z_{\alpha/2} + \frac{\sqrt{n_2}\bmu^\top\hat\bw}{\sqrt{\bmu^\top\bSigma^{-1}\bmu}} \right). \\
\end{aligned}
\end{equation*}
Notice that
\begin{equation*}
\begin{aligned}
   \bmu^\top\hat\bw -  \bmu^\top\bSigma^{-1}\bmu &= \bmu^\top\hat\bw -  \bmu^\top\bw^\star = \bmu^\top(\hat\bw - \bw^\star) \\
   &\leq \|\bmu\|_\infty\|\hat\bw - \bw^\star\|_1 = O\left(\|\bmu\|_\infty s \sqrt{\frac{\log p}{n_1}}\right) \rightarrow 0.
\end{aligned}
\end{equation*}
As a result, we know the asymptotic power is
\[
    \beta(T_{\hat\bw}) = \Phi\left( -z_{\alpha/2} + \sqrt{n_2 \bmu^{\top}\bSigma^{-1}\bmu} \right).
\]

\subsection{Proof of Theorems \ref{thm:exchangeable} and \ref{thm:general_exchangeable}}

Theorem~\ref{thm:exchangeable} is a direct corollary of Theorem~\ref{thm:general_exchangeable} by setting $T_k = p_k$, we only prove Theorem~\ref{thm:general_exchangeable} here.
Conditioning on the observed data $\calD$, we know its random permutations $\calD^{\pi_1}, \calD^{\pi_2}, \dots $ are independent from each other. Therefore, the resultant statistics $T_k = g(\calD^{\pi_k})$ are independent and identically distributed conditioning on $\calD$. By the de Finete theorem \citep{aldous1985exchangeability} which states that a mixture of independent and identically distributed sequences are exchangeable, we know is $(T_1, T_2, \dots)$ is an exchangeable sequence, and hence $(T_1, \dots, T_m)$ is exchangeable for any finite $m$. 

\subsection{Proof of Theorem \ref{thm:Mhatrho}}

According to the central limit theorem for exchangeable random variables, we have
\begin{equation*}
    M_{\rho} = {\bar Z}/{\sqrt{(1 + (m-1)\rho)/m}} \overset{d}{\rightarrow} N(0,1).
\end{equation*}
Let $\hat\rho$ be a consistent estimator for $\rho \neq 0$, i.e., $\hat\rho \overset{p}{\rightarrow} \rho$. Hence,
\begin{equation*}
  \frac{\sqrt{(1 + (m-1)\rho)/m}}{\sqrt{(1 + (m-1)\hat\rho)/m}} \overset{p}{\rightarrow} 1.
\end{equation*}
As a result, we know
\begin{equation*}
    M_{\hat\rho} = {\bar Z}/{\sqrt{(1 + (m-1)\hat\rho)/m}} \overset{d}{\rightarrow} N(0,1).
\end{equation*}

\section{Additional Numerical Results} \label{supp:numeric}
Figures \ref{fig:n100_normal} and \ref{fig:n100_t6} reports the size and power of different tests with $n=100$ for samples following multivariate normal distribution and $t_6$ distribution, respectively. The pattern of size and power for different tests is similar to that of $n=40$. The proposed MPT can control the type I error rate below the pre-specified significance level $\alpha=0.05$ while the CLX test completely fails to control the size. Among those tests which can retain the type I error rate, the proposed MPT is the most powerful one for both CS and AR covariance structure. The studies also reveal that the proposed MPT is quite robust to non-Gaussianity.

\bigskip
\begin{figure}[H]
\centering
    \includegraphics[width=0.95\textwidth]{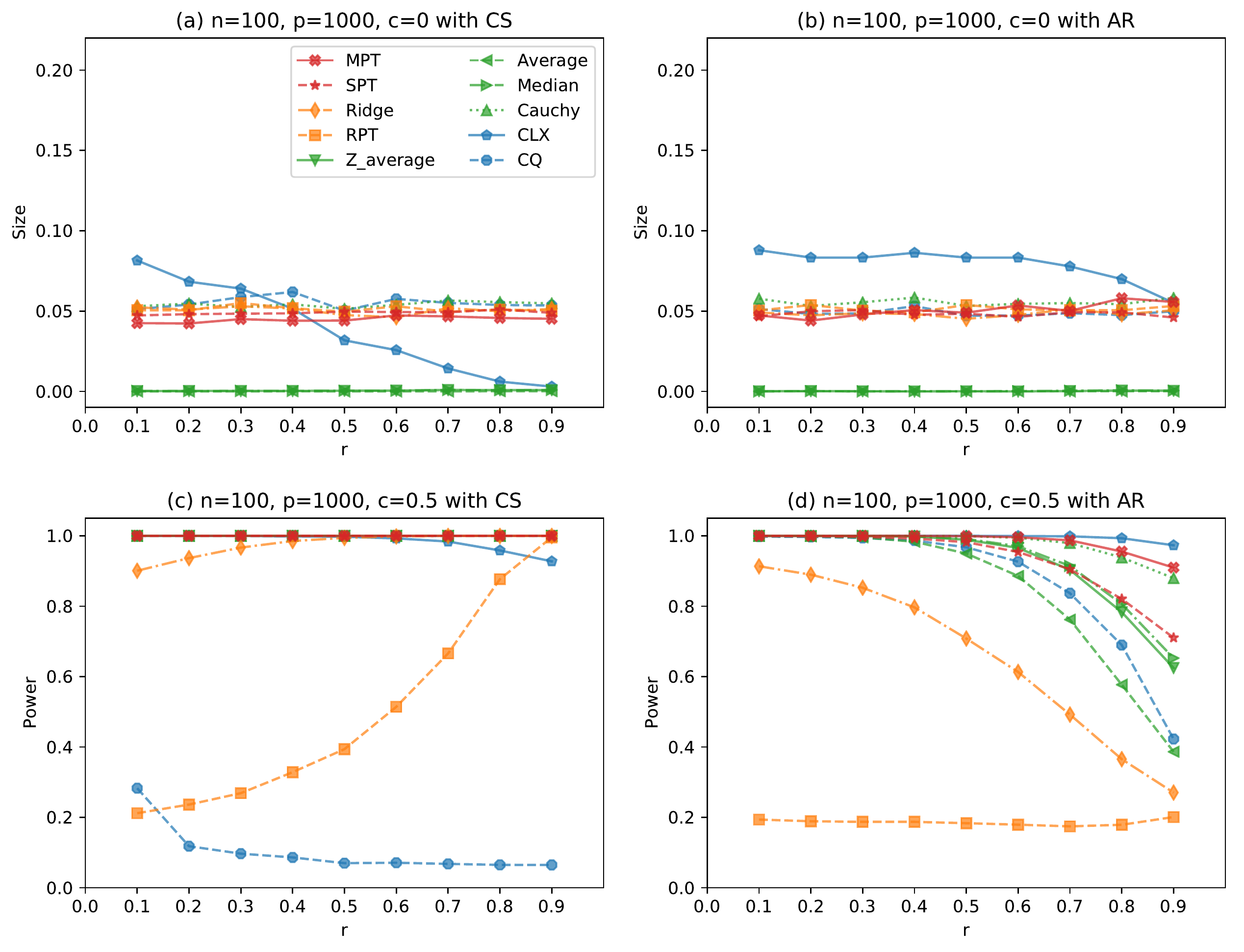}
\caption{Size and power of different tests for normally distributed samples with $n=100$. Panels (a) and (b) show size ($c=0$) under the null hypothesis for the CS and AR structure, respectively. Panels (c) and (d) show power ($c=0.5$) under the alternative hypothesis for the CS and AR structure, respectively.}
\label{fig:n100_normal}
\end{figure}
\newpage
\begin{figure}[H]
\centering
    \includegraphics[width=0.95\textwidth]{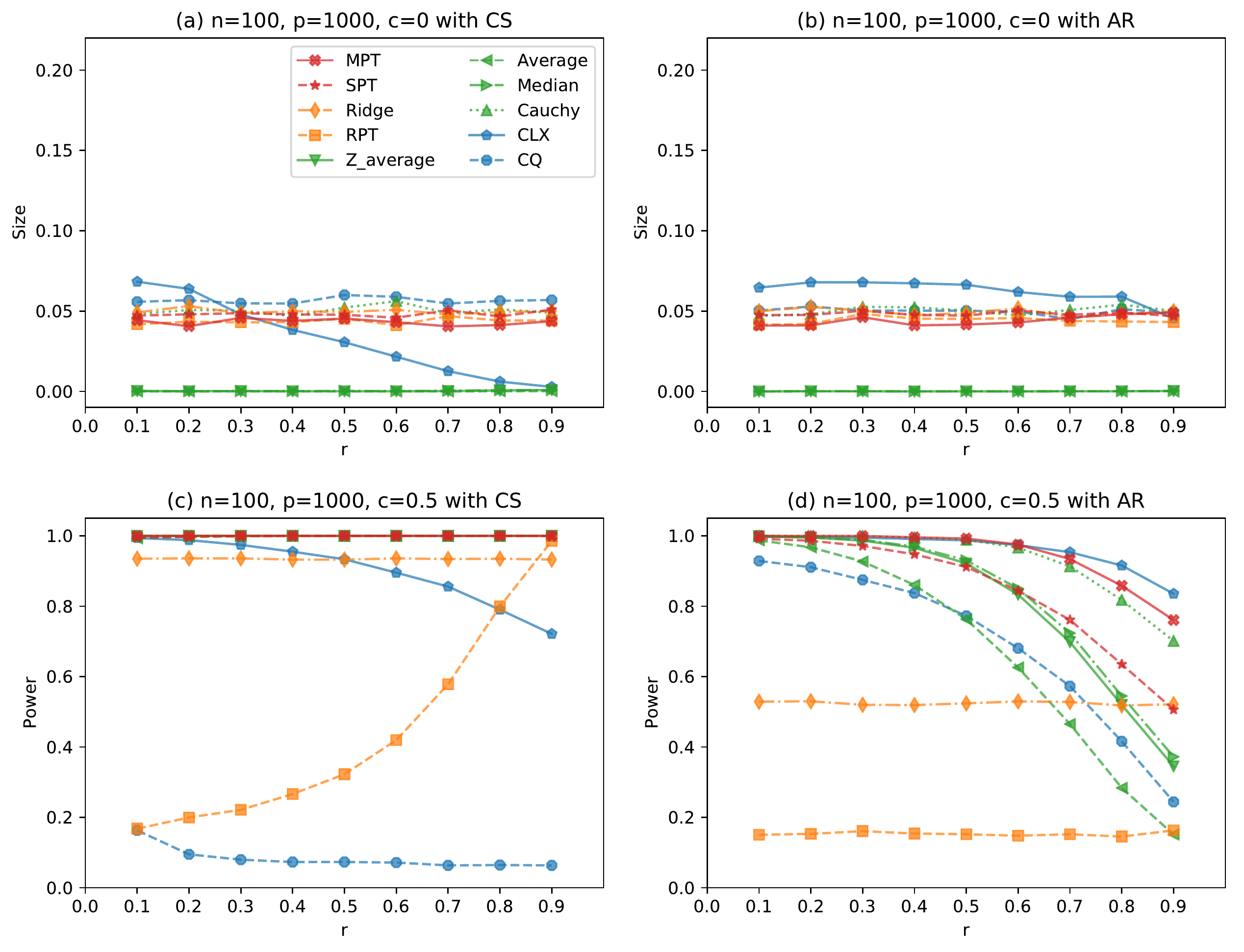}
\caption{Size and power of different tests for samples following multivariate $t_6$ distribution with $n=100$. Panels (a) and (b) show size ($c=0$) under the null hypothesis for the CS and AR structure, respectively. Panels (c) and (d) show power ($c=0.5$) under the alternative hypothesis for the CS and AR structure, respectively.}
\label{fig:n100_t6}
\end{figure}

\end{document}